\documentclass[12pt,a4paper]{amsart}

\let\oldtocsection=\tocsection
\let\oldtocsubsection=\tocsubsection
\let\oldtocsubsubsection=\tocsubsubsection
\renewcommand{\tocsection}[2]{\hspace{0em}\oldtocsection{#1}{#2}}
\renewcommand{\tocsubsection}[2]{\hspace{1em}\oldtocsubsection{#1}{#2}}
\renewcommand{\tocsubsubsection}[2]{\hspace{2em}\oldtocsubsubsection{#1}{#2}}

\usepackage[left=1.6cm, top=1.9cm, bottom=1.9cm, right=1.6cm]{geometry}

\usepackage{amsmath,amssymb,amscd,amsthm,mathrsfs,bm}
\usepackage{float}
\usepackage[colorlinks=true
,urlcolor=blue
,anchorcolor=blue
,citecolor=blue
,filecolor=blue
,linkcolor=blue
,menucolor=blue
,pagecolor=blue
,linktocpage=true
,pdfproducer=medialab
]{hyperref}

\numberwithin{equation}{section}

\theoremstyle{plain}
\newtheorem{theorem}{Theorem}[section]

\newtheorem{proposition}[theorem]{Proposition}

\newtheorem{conjecture}[theorem]{Conjecture}
\theoremstyle{definition}
\newtheorem{example}[theorem]{Example}

\newtheorem{remark}[theorem]{Remark}

\newcommand{\la}{\lambda}
\newcommand{\La}{\Lambda}

\newcommand{\arrow}{\rightarrow}
\newcommand{\map}{\mapsto}

\newcommand{\diff}[2]{\frac{d#1}{d#2}}
\newcommand{\Diff}[3]{\left . \frac{d}{d#2}#1\right |_{#3}}

\newcommand{\alg}{\mathfrak{g}}

\newcommand{\Alg}{\mathbb{A}}

\newcommand{\F}[1]{\mathcal{F}\bra{#1}}

\newcommand{\smf}[1]{\mathcal{C}^\infty(#1)}

\newcommand{\Si}{\mathbb{S}^1}
\newcommand{\Cm}{\mathbb{C}}

\newcommand{\Km}{\mathbb{K}}

\newcommand{\Tm}{\mathbb{T}}

\newcommand{\T}{{\rm T}}

\newcommand{\Fun}{\mathscr{F}}
\renewcommand{\F}{\mathcal{F}}
\renewcommand{\H}{\mathcal{H}}
\renewcommand{\P}{\mathcal{P}}

\newcommand{\Liea}{\mathscr{L}}

\newcommand{\g}{\tilde{g}}
\newcommand{\h}{\tilde{h}}
\newcommand{\f}{\tilde{f}}

\newcommand{\q}{\bm{q}}
\newcommand{\ff}{\bm{f}}
\renewcommand{\gg}{\bm{g}}

\newcommand{\uf}{\bm{u}}

\newcommand{\bg}{\boldsymbol{\gamma}}

\newcommand{\pr}{\partial}
\newcommand{\me}{\geqslant}
\newcommand{\les}{\leqslant}
\newcommand{\bra}[1]{\! \left (#1\right )}
\newcommand{\brac}[1]{\! \left [#1\right ]}
\newcommand{\bc}[1]{[\![#1]\!]}
\newcommand{\pobr}[1]{\left \{#1\right \}}
\newcommand{\ddual}[1]{\bigl \langle #1 \bigr \rangle}
\newcommand{\dual}[1]{\langle #1 \rangle}
\newcommand{\var}[2]{\frac{\delta #1}{\delta #2}}
\newcommand{\pd}[2]{\frac{\partial #1}{\partial #2}}

\newcommand{\pmatrx}[1]{\begin{pmatrix} #1 \end{pmatrix}}

\DeclareMathOperator{\const}{const}
\DeclareMathOperator{\ad}{ad}

\DeclareMathOperator{\diag}{diag}

\newcommand{\abs}[1]{\left |#1\right |}

\newcommand{\Vect}[1]{{\rm Vect}(#1)}

\newcommand{\eqq}[1]{\begin{align*} #1 \end{align*}}
\newcommand{\eq}[1]{\begin{align} #1 \end{align}}

\newcommand{\eps}{\varepsilon}

\usepackage{enumitem}
\setlist{nolistsep}

\begin{document}

\title[Novikov algebras and a classification of multicomponent Camassa-Holm equations]{Novikov algebras and a classification of multicomponent Camassa-Holm equations}

\author{Ian A.B. Strachan and B\l a\.zej M. Szablikowski}

%\address{}
\address{School of Mathematics and Statistics\\ University of Glasgow\\Glasgow G12 8QQ\\ U.K.}
\email{ian.strachan@glasgow.ac.uk\bigskip}
%\address{}
\address{Faculty of Physics, Adam Mickiewicz University, Umultowska 85, 61-614 Pozna\'n, Poland}
\email{bszablik@amu.edu.pl}

 \subjclass{}

\date{}

\begin{abstract}
A class of multi-component integrable systems  associated to Novikov algebras, which interpolate between KdV and Camassa-Holm type equations, is obtained. The construction is based on the classification of low-dimensional Novikov algebras by Bai and Meng. These multi-component bi-Hamiltonian systems obtained by this construction may be interpreted as Euler equations on the centrally extended Lie algebras associated to the Novikov algebras. The related bilinear forms generating cocycles
of first, second and third order are classified. Several examples, including known integrable equations, are presented.
\end{abstract}

\maketitle

{\small \parskip0.1ex \tableofcontents}
%\tableofcontents

\newpage
\section{Introduction}

The Camassa-Holm equation \cite{CH}
\begin{equation}
v_t - v_{xxt} = \alpha v_{x} -3 v v_x + 2 v_x v_{xx} + v v_{xxx} \,,
\label{CHeqn}
\end{equation}
an example of a $(1+1)$-dimensional integrable system, has many intriguing mathematical properties, amongst them being:
\begin{itemize}
\item[$\bullet$] the existence of multi \lq peakon\rq~solutions;
\item[$\bullet$] the non-existence of a $\tau$ function or functions.
\end{itemize}
Its other properties have more in common with equations such as the KdV equation: the existence of a Lax pair, solvability via the inverse scattering transform,
and a bi-Hamiltonian structure. In fact, the bi-Hamiltonian structure, and hence the Camassa-Holm itself, may be found by exploiting the tri-Hamiltonian structure
of the KdV hierarchy \cite{FF,F,OR}. The component parts of the bi-Hamiltonian pair for the KdV equation, namely
\begin{equation}
\frac{d~}{dx}\,, \qquad \frac{d^3~}{dx^3}\,, \qquad u \frac{d~}{dx} + \frac{1}{2} u_x
\label{kdvparts}
\end{equation}
are pair-wise compatible and hence may be recombined to form the bi-Hamiltonian structures
\begin{eqnarray*}
\displaystyle{\mathcal{P}_1} & = & \displaystyle{\frac{d~}{dx} - \frac{d^3~}{dx^3}\,,}\\
\displaystyle{\mathcal{P}_2} & = & u \displaystyle{\frac{d~}{dx} + \frac{1}{2} u_x}
\end{eqnarray*}
and applying the Lenard-Magri recursion scheme results in the Camassa-Holm equation (\ref{CHeqn}), where $u=v-v_{xx}\,$.  Another common feature between Camassa-Holm and KdV equations, closely related
to the above bi-Hamiltonian structure, is that both systems may be written as Euler equations on the Virasoro algebra (see \cite{KM} and references therein).

The class of Hamiltonian operators such that this tri-Hamiltonian duality may be most easily applied was first derived by Balinskii and Novikov \cite{BN} as
a special case of the Dubrovin-Novikov operators of hydrodynamic type \cite{DN}. The conditions for the operator
\[
\P^{ij} = { c^{ij}_{k}} u^k \diff{}{x} + { b^{ij}_{k} } u^k_x\,,  \qquad c^{ij}_{k} = b^{ij}_{k}+b^{ji}_{k}\,,
\]
to be a Hamiltonian  places purely algebraic conditions on the constants $b^{ij}_{k}$, and the
corresponding algebraic structure - on regarding these constants as structure constants - is now known as a Novikov algebra. With each Novikov algebra there is associated a translationally invariant Lie algebra, which can be centrally extended. Condition for the existence of cocycles
(either first-order or third-order Gelfand-Fuks cocycles) is equivalent to the existence of symmetric bilinear forms $g$ and $h$ on the Novikov algebra which satisfy
certain (quasi-Frobenius and Frobenius, respectively) compatibility conditions. Second order cocycles result in antisymmetric bilinear forms which again satisfy certain algebraic
relations.
This construction is outlined in Section~\ref{sec2}. Thus one obtains
a multi-component tri-Hamiltonian structure
\[
g^{ij}\frac{d~}{dx}\,, \qquad h^{ij}\frac{d^3~}{dx^3}\,, \qquad  \P^{ij}
\]
in direct analogue to (\ref{kdvparts}), defined algebraically in terms of Novikov algebras and compatible symmetric bilinear forms on the algebra. Using these ideas, the centrally extended translationally invariant Lie algebras associated to Novikov algebras can be considered as multi-component linear generalization of the Virasoro algebra.

A two-component integrable generalization of the Camassa-Holm equation, called CH2, has been proposed in \cite{LZ,CLZ}. This new system admits a Lax pair \cite{CLZ}, of the same type as the original Camassa-Holm equation, which is connected to the energy-dependent Schr\"odinger spectral problem, see \cite{AF1,AF2}.  An alternative Lie algebraic approach for the derivation of CH2, which is closer to the one considered in this article, was presented in \cite{Fa}. Further, taking the advantage of the energy-dependent Schr\"odinger spectral problems the CH2 equation is generalized in \cite{HI} to produce
an integrable multi-component family CH$(n,k)$ of equations with $n$ components and $k$ velocities. The direct relationships between our results with these from \cite{HI} remains to be clarified (see Section~\ref{exx} and Example~\ref{ex64} for some specific examples).

The purpose of this paper is two-fold. By mirroring the construction of the Camassa-Holm equation we will obtain multi-component versions of this equation (in fact, by
splitting the various structures more finely one can obtain equations which interpolate between KdV and Camassa-Holm equations). The nonlinear terms in these equations
are controlled by the properties of the Novikov algebra and associated bilinear forms. Secondly, using the (low-dimensional) classification of Novikov algebras
 obtained by Bai and Meng \cite{BM1} - extended to classify the second and third order Gelfand-Fuks cocycles - one can obtain a classification scheme for equations in this
 class purely in terms of an underlying algebraic structure. It should be noted that these algebras are classified over $\mathbb{C}\,$; in some cases a finer
decomposition exists over $\mathbb{R}\,$: see, for example \cite{BG}. It turns out that many Novikov algebras result in highly degenerate systems of integrable equations, so the classification scheme has to be extended to rule out these degenerate cases.

\section{Novikov algebras}\label{sec2}

We begin by briefly reviewing some of the basic concepts of linear Poisson tensors and brackets that will be needed elsewhere in the paper.
The material may all be found in the original papers \cite{BN,DN} or in expositions such as \cite{Dorf}. Other basic results may be found in the
Appendix.

\subsection{Linear Poisson tensors of hydrodynamic type}

Following \cite{BN} consider a homogeneous first order $n\times n$ operator\footnote{The summation convention is used throughout this article.}
\eq{\label{po}
\P^{ij} = g^{ij}(u) \diff{}{x} + b^{ij}_k u^k_x\,,\qquad x\in\Si,
}
depending on fields $u^1(x),\ldots,u^n(x)$. Here,  $g^{ij}(u) = c^{ij}_k u^k $ is symmetric and $b^{ij}_k$, $c^{ij}_{k}$ are constants. The operator is a Poisson operator if and only if
\begin{itemize}
\item $c^{ij}_k = b^{ij}_k + b^{ji}_k\,$;
\item $b^{ij}_k$ is the set of structure constants of an algebra $\Alg$, that is $e^i\cdot e^{j} = b^{ij}_k\, e^k$ where $e^1,\ldots,e^n$ are basis vectors, such that
\begin{subequations}\label{cond}
\eq{\label{conda}
&(a\cdot b)\cdot c = (a\cdot c)\cdot b\,,\\
&(a\cdot b)\cdot c - a\cdot(b\cdot c) = (b\cdot a)\cdot c - b\cdot(a\cdot c)\,.
}
\end{subequations}
\end{itemize}

This structure $\Alg$ is called a Novikov algebra.\footnote{All structures are considered over the field of complex numbers $\Cm$.} The relations \eqref{cond} may be written in a simple way by defining left and right multiplications by $L_a b = R_b a= a\cdot b$.
With these the identities \eqref{cond} are equivalent to
\eqq{
	\brac{R_a,R_b} =0\,,\qquad \brac{L_a,L_b} = L_{[a,b]}\,,
}
respectively. Here $\brac{a,b} = a\cdot b - b\cdot a$. Novikov Lie algebras $\Alg$ are Lie admissible, that is
the commutator defines structure of a Lie algebra on the underlying vector spaces $\Alg$. Note that if the multiplication is
commutative then the Novikov conditions \eqref{cond} reduce to associativity conditions.

Let $\Alg^*$ be the dual algebra with respect to the standard pairing $(\,,\,):\Alg^*\times\Alg\arrow\Cm$.
Then $(L_a^*u,b) = (R_b^*u,a) := (u,a\cdot b)$, where $a,b\in\Alg$ and $u\in\Alg^*$. The following
relation, which will be needed later, is equivalent to \eqref{conda}:
\eq{\label{rel1}
R_bL_a c = L_{R_ba}c \quad\iff\quad L^*_aR^*_b u = L^*_{R_ba}u.
}

\subsection{Lie-Poisson brackets}

Consider the infinite-dimensional Lie algebra $\Liea_\Alg$ on the space of $\Alg$-valued functions
of $x\in\Si$, with a Lie bracket of the form
\eq{\label{lb}
	\bc{a,b} : = a_x\cdot b - b_x\cdot a\,,\qquad a_x\equiv \diff{a}{x}.
}
In fact, \eqref{lb} defines a Lie bracket  if and only if the multiplication $\cdot:\Alg\times\Alg\arrow\Alg$ satisfies the conditions \eqref{cond} and hence the algebra $\Alg$ is a Novikov algebra.

The Poisson bracket associated to the Poisson operator \eqref{po} is a Lie-Poisson bracket associated to the Lie algebra $\Liea_\Alg$ (see Appendix~\ref{liepo}):
\eq{\label{pb}
\pobr{\H, \mathcal{F}}[u]:= \int_{\Si} \var{\mathcal{F}}{u^i}\P^{ij}\var{\H}{u^j}\,dx \equiv \ddual{u,\bc{\delta_u \mathcal{F},\delta_u \H}}\,,\qquad u\in\Liea_\Alg^*,
}
where $\H,\mathcal{F}\in\Fun(\Liea_\Alg^*)$ are functionals on the (regular) dual space to $\Liea_\Alg$. Here, the duality pairing $\Liea_\Alg^*\times\Liea_\Alg\arrow\Cm$ is given by
\eqq{
	\dual{u,a} := \int_{\Si}\bra{u,a} dx\,,
}
where $u\in\Liea_\Alg^*$ and $a\in\Liea_\Alg$.
Functionals from $\Fun(\Liea_\Alg^*)$  have the form
\eqq{
\H[u]=\int_{\Si} H(u,u_x,u_{xx},\ldots)\, dx\,,
}
their (variational) differentials $\delta_u \H = \var{\H}{u^i}e^i$ belong to $\Liea_\Alg\,,$
which follows from the relation \eqref{differ}.
The coadjoint  action, such that \eqref{coad}, is
\eqq{
\ad^*_au = -(R^*_a u)_x - L^*_{a_x}u.
}

\subsection{Central extensions}\label{cene}

Centrally extended Poisson bracket \eqref{pb} is given by
\eq{\label{epb}
\pobr{\H, \F}[u]:= \ddual{u,\bc{\delta_u \H,\delta_u \F}} + \omega\bra{\delta_u \H,\delta_u \F},
}
where the bilinear map $\omega$ on $\Liea_\Alg$ is a $2$-cocycle, see Appendix~\ref{central}. Let $\omega$ be defined by means of  a $1$-cocycle $\phi:\Liea_\Alg \arrow \Liea_\Alg^*$, i.e. the relation \eqref{op} is valid.
Then, the Poisson operator $\P$ of \eqref{epb}, such that $\pobr{\H, \F}:= \ddual{\P\delta_u \H,\delta_u \F}$,
has the form
\eq{\label{epo}
		\P\gamma = - \ad^*_{\gamma}u + \phi (\gamma)\,,\qquad \gamma\in\Liea_\Alg.
}

We will consider now some differential $1$-cocycles, yielding central extensions of Lie algebras associated with Novikov algebras.  They are generated by appropriate bilinear forms satisfying various algebraic conditions and were originally derived in~\cite{BN}.

\begin{itemize}
\smallskip
\item{\underline{Order-one cocycles:}}
A symmetric bilinear form $g$ on $\Alg$, $g(a,b)\equiv (\g(a),b)$, defines $1$-cocycle of order one
$\phi = \g\diff{}{x}$, that is the associated $2$-cocycle is
\eqq{
	\omega_1(a,b) = \int_{\Si} g(a_x,b)\, dx\,,
}	
if and only if the quasi-Frobenius condition
\eq{\label{cond1}
g(a\cdot b,c) = g(a,c\cdot b)
}
holds for all $a,b,c\in\Alg$. The condition \eqref{cond1} is equivalent either to
\eq{\label{rel2}
\g(R_b a) = R^*_b\g(a),
}
or
\eq{\label{rel3}
L^*_a \g(b) = L^*_b \g(a).
}

From the more general formalism of Poisson tensors of hydrodynamic type \cite{DN} we know that
the bilinear form $g^{ij}(u)=c^{ij}_{k} u^k + g_0^{ij}$, where $g_0$ is generating the above first order cocycle, if $\det g^{ij}(u)\neq 0$, can be interpreted as a contravariant flat metric.
Thus this bilinear form may be thought about in two different ways: as an inhomogeneous flat metric, or as a homogeneous flat metric with the addition of
a first-order cocycle term. Even if $\det(c^{ij}_{~~k} u^k)=0$ it still defines a Poisson tensor \cite{Grinberg}.

\smallskip

\item{\underline{Order-two cocycles:}}
A skew-symmetric  bilinear form $f$ on $\Alg$, $f(a,b)\equiv (\tilde{f}(a),b)$, defines $1$-cocycle of order two
$\phi = \f\diff{^2}{x^2}$, with the associated $2$-cocycle
\eqq{
	\omega_2(a,b) = \int_{\Si} f(a_{xx},b)\, dx\,,
}	
if and only if
\eq{\label{cond3}
		f(a\cdot b, c) = f(a, c\cdot b)
}
and
\eq{\label{cond4}
		f(a\cdot b, c) + f(b\cdot c, a) + f(c\cdot a, b) = 0\,,
}
that is the quasi-Frobenius and cyclic conditions are satisfied. The condition \eqref{cond3} may be rewritten in two equivalent ways: as
\eq{\label{rel5}
\f(R_b a) = R^*_b\f(a)\,,
}
or as
\eq{\label{rel6}
L^*_a \f(b) = -L^*_b \f(a).
}
On the other hand \eqref{cond4} is equivalent to
\eqq{
		L^*_b\f(a) = \f(a\cdot b - b\cdot a).
}

\smallskip

\item{\underline{Order-three cocycles:}}
A symmetric bilinear form $h$ on $\Alg$, $\h(a,b)\equiv (\tilde{h}(a),b)$, defines $1$-cocycle of third order
$\phi = \h\diff{^3}{x^3}$, with the associated $2$-cocycle
\eqq{
	\omega_3(a,b) = \int_{\Si} h(a_{xxx},b)\, dx\,,
}	
if and only if the tensor $h(a\cdot b,c)$ is totally symmetric, which reduces to two conditions:
\eq{\label{cond2}
		h(a\cdot b,c) = h(a,c\cdot b),\qquad  h(a,b\cdot c) = h(a,c\cdot b).
}
The conditions \eqref{cond2} are equivalent to
\eq{\label{rel4}
	L^*_a\h(b) = R^*_a\h(b) = \h(a\cdot b)=\h(b\cdot a).
}
\end{itemize}

\medskip

\noindent There are no higher order cocycles with constants coefficients. Let us denote by
$\mathscr{C}^i_{\Alg}$ the linear space spanned by bilinear forms generating cocycles
of $i$-th order and associated with a Novikov algebra~$\Alg$.

\section{Multicomponent bi-Hamiltonian Camassa-Holm hierarchies}

Consider the following pair of compatible Poisson operators $\P_0$ and $\P_1$ on $\Liea_\Alg$, associated with some Novikov algebra $\Alg\,$:
\begin{subequations}\label{pair}
\eq{
	\P_1\gamma = (R^*_\gamma u)_x + L^*_{\gamma_x}u + \g_1 \gamma_x +  \f_1 \gamma_{xx} + \h_1 \gamma_{xxx}
}
and
\eq{
	\P_0\gamma = \g_0 \gamma_x + \f_0 \gamma_{xx} + \h_0 \gamma_{xxx},
}
\end{subequations}
where $u\in\Liea^*_\Alg$ and $\gamma\in\Liea_\Alg$. Here $g_0$ and $g_1$ are symmetric bilinear forms on $\Alg$  generating
first order $1$-cocycles, $f_0$ and $f_1$  are skew-symmetric bilinear forms generating $1$-cocycles of
order two, while $h_0$ and $h_1$ are symmetric bilinear forms generating third order $1$-cocycles.
The operator $\P_1$  has the form of a centrally extended Poisson operator \eqref{epo}. The compatibility  of  $\P_0$ and $\P_1$ is a consequence of the fact that any linear combination of $1$-cocycles is itself
a $1$-cocycle. In many research papers the second Poisson structure is obtained by freezing the first one, which often yields restricted class of compatible Poisson operators. In the case of the linear Poisson operator $\P_1$ our approach gives the most general compatible with it Poisson operator $\P_2$ with constant
coefficients.

We can formulate the bi-Hamiltonian chain associated to the above Poisson  pair in the form
\eq{
\label{uh}
\begin{split}
u_{t_0} &= \P_0 \delta_u \H_0 \equiv 0\,,\\
u_{t_1} &= \P_1 \delta_u \H_0  = \P_0 \delta_u \H_1\,,\\
u_{t_2} &= \P_1 \delta_u \H_1 = \P_0 \delta_u \H_2\,,\\
&\qquad \vdots\,,\\
\end{split}
}
where $\H_i\in\Fun\bra{\Liea_\Alg^*}$ and $\H_0$ is a Casimir of $\P_0$, $t_i$ are evolution parameters (times)
of respective flows associated to vector fields $u_{t_i}$ on $\Liea_\Alg^*$.

Consider the operator $\Lambda:\Liea_\Alg\arrow\Liea_\Alg^*$ defined by
\eq{\label{lam}
\Lambda\gamma:= \g_0 \gamma + \f_0 \gamma_x + \h_0 \gamma_{xx},
}
such that $\P_0\gamma\equiv \Lambda \gamma_x$. The operator $\Lambda$ is self-adjoint,
i.e. $\Lambda^\dag= \Lambda$, and when $\Lambda$ is invertible it can be interpreted as an inertia operator,
see Remark~\ref{rem1}.
We assume that $\Lambda$ as a map is a diffeomorphism, and hence we impose its invertibility by the condition that the bilinear form $g_0$ is
nondegenerate (see Remark~\ref{rem}). As we will see in the following theorem, to obtain local forms of evolution equations one must perform the change of coordinates
\eq{\label{trans}
v:=\Lambda^{-1} u,
}
which is of (linear) Miura-type. Thus the hierarchy \eqref{uh} transforms into
\eq{\label{vh}
\begin{split}
v_{t_0} &= \tilde{\P}_0 \delta_v \H_0 \equiv 0\,,\\
v_{t_1} &= \tilde{\P}_1 \delta_v \H_0  = \tilde{\P}_0 \delta_v \H_1\,,\\
v_{t_2} &= \tilde{\P}_1 \delta_v \H_1 = \tilde{\P}_0 \delta_v \H_2\,,\\
&\qquad \vdots\,,\\
\end{split}
}
where $\tilde{\P}_i = \Lambda^{-1}\P_i(\Lambda)^{-1}$ and $\delta_v\H_j = \La\delta_u\H_j$.

\begin{theorem}\label{thr}
The first two evolution equations from the hierarchy \eqref{vh} are
\eqq{
	v_{t_1} &= v_x\cdot c\,,\\
\g_0(v_{t_2}) + \f_0(v_{xt_2}) + \h_0(v_{xxt_2}) &=
\g_0\bra{v_x\cdot(v\cdot c)} + \g_0\bra{v\cdot(v_x\cdot c)} + L^*_{v\cdot c} \g_0(v_x)\\
&\quad + \f_0\bra{v_x\cdot(v_x\cdot c)} + \f_0\bra{v_{xx}\cdot(v\cdot c)} \\
&\quad+ 2\h_0\bra{(v_x\cdot c)\cdot v_{xx}} + \h_0\bra{(v\cdot c)\cdot v_{xxx}}\\
&\quad +\g_1(v_x\cdot c) +\f_1(v_{xx}\cdot c) + \h_1(v_{xxx}\cdot c)\,,
}
where $c=\const\in\Liea_\Alg$ and $\delta_u\H_0 = c$. The densities of the first three conserved quantities (Hamiltonians), such that
\eqq{
\H_i[v] = \int_{\Si} H_i(v,v_x,v_{xx},\ldots)\, dx\,,\qquad i=0,1,\ldots
}
are
\eqq{
	H_0 &= g_0(c,v)\,,\\
	H_1 &= \frac{1}{2}\,g_0(v,v\cdot c) + \frac{1}{2}\,f_0(v_x,v\cdot c) + \frac{1}{2}\,h_0(v_{xx},v\cdot c)\,,\\
	H_2 &= \frac{1}{3}\,g_0\bra{v,v\cdot(v\cdot c)} + \frac{1}{3}\,f_0\bra{v_x,v\cdot(v\cdot c)}
	+ \frac{1}{3}\,h_0\bra{v\cdot c,v\cdot v_{xx}}\\
&\quad +\frac{1}{6}\,g_0\bra{v\cdot c,v\cdot v} + \frac{1}{6}\,h_0\bra{v_x\cdot c,v_x\cdot v}\\
&\quad+ \frac{1}{2}\,g_1(v,v\cdot c) + \frac{1}{2}\,f_1(v_x,v\cdot c) + \frac{1}{2}\,h_1(v_{xx},v\cdot c)\,.
}
\end{theorem}
\begin{proof}
Recall that by Poincar\'e lemma a closed $1$-form (co-vector) $\gamma$ in a star-shape (local) coordinate system
$\{u\}$ is exact, that is $\gamma = \delta_u \H$. The functional $\H$ can be obtained
using the homotopy formula~\cite{O}:
\eqq{
		\H[u] = \int_0^1\dual{u, \gamma(\la u)} d\la\,.
}
A useful relation, which follows from \eqref{rel2}, \eqref{rel5} and \eqref{rel4}, is
\eqq{
	R^*_a\La = \La R_a\qquad\iff\qquad \La^{-1}R^*_a = R_a\La^{-1}\,,\qquad a\in\Alg\,,
}
which will be used below.

Casimirs of $\P_0$ are constants $c\in\Alg$. Hence, taking $\delta_u \H_0 = c$ we have
\eqq{
		\H_0 = \dual{u,c} = \dual{\La v,c}.
}
Thus,
\eqq{
	u_{t_1} = \P_1\delta_u\H_0 = R^*_c u_x
}
and from \eqref{uh} one finds that $\delta_u\H_1 = \La^{-1}R^*_c u$. Using the homotopy formula one finds that
\eqq{
		\H_1 = \frac{1}{2}\dual{\La^{-1}u, R^*_c u} = \frac{1}{2}\dual{\La v, v\cdot c}.
}
The second flow is
\eq{\label{secflow}
\begin{split}
	u_{t_2} &= \P_1\delta_u\H_1\\
	&= (R^*_{R_c\La^{-1} u} u)_x + L^*_{\La^{-1}u_x}R^*_cu + \g_1 R_c\La^{-1}u_x + \f_1 R_c\La^{-1} u_{xx}
	+ \h_1 R_c\La^{-1} u_{xxx}
\end{split}
}
or equivalently
\eqq{
	\La v_{t_2} =
	(R^*_{R_cv} \La v)_x + L^*_{v_x}\La R_c v +\g_1 R_cv_x + \f_1 R_c v_{xx} + \h_1 R_cv_{xxx}.
}
Note that
\eq{\label{rr}
	L^*_{v_x}\f_0 R_c v_x = 0,
}
which is a consequence of relations \eqref{rel1}, \eqref{rel5} and \eqref{rel6}.
Using the properties of the symmetric forms $\g_0$, $\h_0$ and \eqref{rr} one can show that
\eqq{
	L^*_{v_x}\La R_c v = \frac{1}{2}\bra{L^*_v\g_0 R_c v + L^*_{v_x}\h_0 R_c v_x}_x.
}
Hence,
\eqq{
	\delta_v\H_2\equiv \La\delta_u\H_2 = R^*_{R_cv} \La v + \frac{1}{2} L^*_v\g_0 R_c v  + \frac{1}{2}L^*_{v_x}\h_0 R_c v_x\,,
	+ \g_1 R_cv + \f_1 R_cv_x + \h_1 R_cv_{xx}
}
and finally, using the homotopy formula, one obtains
\eqq{
	\H_2 &= \frac{1}{3}\dual{\La v, R_{R_cv}v} + \frac{1}{6}\dual{\g_0R_cv,L_vv}
	+ \frac{1}{6}\dual{\h_0R_cv_x,L_{v_x}v}\\
	&\quad + \frac{1}{2}\dual{\g_1R_cv,v} + \frac{1}{2}\dual{\f_1R_cv_x,v} + \frac{1}{2}\dual{\h_1R_cv_{xx},v},
}
which finishes the proof.
\end{proof}

For non-nilpotent Novikov algebras the invertibility of $\Lambda$ guarantees the existence of
the (hereditary) recursion operator $\mathcal{R} = \P_1 \P_0^{-1}$, such that $u_{t_i} = \mathcal{R}^i u_{t_0}$, and hence the existence of the infinite hierarchies of commuting evolution equations and conserved quantitates. Thus the equations from the hierarchy \eqref{uh} (or equivalently \eqref{vh}) are integrable, see \cite{O}.

\begin{remark}\label{rem1}
Observe that the Hamiltonian flow \eqref{secflow} on the dual space $\Liea_\Alg^*$ is the Euler equation
(see Appendix~\ref{Euler}),
corresponding to the centrally extended Lie algebra $\Liea_\Alg\oplus\mathscr{C}^1_\Alg\oplus\mathscr{C}^2_\Alg\oplus\mathscr{C}^3_\Alg$, with the quadratic Hamiltonian
\eqq{
		\H_1 = \frac{1}{2}\dual{u, \La^{-1}R^*_c u}.
}
This Euler equation transformed through \eqref{trans} to $\Liea_\Alg$ is the second flow from Theorem~\ref{thr}.
\end{remark}

\medskip

These formulae simplify if the Novikov algebra $\Alg$ has a right-unity $e$, i.e.~$u\cdot e = u$ for all $u\in\Alg$ (recall that if the algebra
has a left-unity then it is automatically commutative and associative). Then, taking $c=e$
the evolution equations and Hamiltonians from Theorem~\ref{thr} simplify to
\eq{
\notag	v_{t_1} &= v_x\,,\\
\label{eq} \g_0(v_{t_2}) + \f_0(v_{xt_2}) + \h_0(v_{xxt_2}) &=
	\Bigl ( \g_0\bra{v\cdot v}  +\frac{1}{2} L^*_{v}\, \g_0\, v+ \f_0\bra{v_x\cdot v}
	+ \frac{1}{2}\h_0\bra{v_x\cdot v_x}\\
\notag	&\qquad + \h_0\bra{v\cdot v_{xx}} +\g_1 v +\f_1 v_x + \h_1 v_{xx}\Bigl )_x
}
and
\eqq{
	H_0 &= g_0(e,v)\,,\\
	H_1 &= \frac{1}{2}\,g_0(v,v) + \frac{1}{2}\,f_0(v_x,v) + \frac{1}{2}\,h_0(v_{xx},v)\,,\\
	H_2 &= \frac{1}{2}\,g_0\bra{v,v\cdot v} + \frac{1}{3}\,f_0\bra{v_x,v\cdot v}
	+ \frac{1}{3}\,h_0\bra{v ,v\cdot v_{xx}} + \frac{1}{6}\,h_0\bra{v,v_x\cdot v_x}\\
&\quad+ \frac{1}{2}\,g_1(v,v) + \frac{1}{2}\,f_1(v_x,v) + \frac{1}{2}\,h_1(v_{xx},v)\,.
}

\begin{remark}\label{rem}
Note that the invertibility of the inertia operator \eqref{lam}, that is $\Lambda = \g  + \f \pr_x + \h \pr_{x}^2$ (for simplicity the subscript is omitted here), is understood in the sense of invertibility of a pseudo-differential
operator \cite{O}.\footnote{There is also possible other  approach to the invertibility of the inertia operator
such as the analytic approach presented in \cite{KLM}.} That is, if $\det\g\neq 0\,,$ the inverse is given by a formal differential series
\eqq{
\Lambda^{-1} = \g^{-1}  -  \g^{-1}\f\g^{-1}\pr_x + \g^{-1}\bra{\f\g^{-1}\f - \h}\g^{-1}\pr_x^2 + \ldots\ .
}
There are two important special cases which will occur in the next section, and for which Theorem~\ref{thr} still holds.
The first one is when $\g=0$ and $\det\f\neq 0$. In this case the inverse is
\eqq{
\Lambda^{-1} = \f^{-1}\pr_x^{-1}  -  \f^{-1}\h\f^{-1} + \f^{-1}\h\f^{-1}\h\f^{-1}\pr_x + \ldots\ .
}
The second is when $\g=\f=0$ and $\det\h\neq 0$. In this case the inverse is $\Lambda^{-1} = \h^{-1}\pr_x^{-2}$.
\end{remark}

\section{A classification of integrable systems on Novikov algebras}\label{cis}

In \cite{BM1} Novikov algebras $\Alg$ of dimension up to $n=3$ were classified over $\mathbb{C}$ and in \cite{BM2} this work was
extended to classify four-dimensional transitive Novikov algebras\footnote{A Novikov algebra is called
transitive (or right-nilpotent) if every $R_a$ is nilpotent.}. In \cite{BG} a
different approach to the classification of low-dimensional Novikov algebras $\Alg$ was presented. This is based on the classification
of the associated Lie algebras $(\Alg,[\, ,\,]), $\footnote{Naturally, these Lie algebras should
not be confused with the translationally invariant Lie algebras $\Liea_\Alg$.} defined by means of a commutator
$\brac{a,b}=a\cdot b - b\cdot a$.
As result the authors classified four-dimensional Novikov algebras for which associated Lie algebras are nilpotent.
However, the full classification of Novikov algebras of dimension four is far from being complete.

To be able to compute the associated evolution equations, using Theorem~\ref{thr}, one must additionally
classify the bilinear forms generating central extensions of the Lie algebras $\Liea_\Alg$.
This classification is presented in Appendix~\ref{class}. The various defining relations for these bilinear forms
are just linear equation and hence may be solved using, for example, Mathematica.
However, not all of these Novikov algebras will lead
to the construction of \lq proper\rq~evolution equations:
\begin{itemize}
\item All Novikov algebras that are direct sums of lower dimensional algebras can be omitted as they lead to evolution equations consisting of decoupled systems associated to the
lower dimensional Novikov algebras.
\item We are interested in the construction of evolution equations for which the coefficients
of associated (infinite-dimensional) vector fields will depend on all fields. This leads to the exclusion of
Novikov algebras for which the generic right multiplication (or equivalently left multiplication) has rank lower then the
algebra dimension.\footnote{One could claim that some of this cases can lead to proper reduced systems. However, careful inspection shows that  for such low-dimensional Novikov algebras this cannot be the case.}
This occurs for all transitive Novikov algebras appearing in the classification schemes in \cite{BM1,BM2,BG}.
\item The invertibility of the inertia operator $\Lambda$ impose nondegeneracy of the bilinear form $g_0$.
Hence, the Novikov algebras for which \lq generic\rq~bilinear forms $g$, generating central extension of
first order, are degenerate can also be omitted.
\end{itemize}

As result, the only relevant Novikov algebras are in dimension one: the field of complex numbers~$\Cm$;
in dimension two:  $(N3)$--$(N6)$; in dimension three: $(C6)$, $(C8)$, $(C9)$, $(C16)$, $(C19)$, $(D2)$--$(D5)$;
in dimension four (within the Novikov algebras from \cite{BM2,BG} and this paper):
$\widetilde{A}_{3,3}$, $\widetilde{A}_{3,4}$, $N^{\mathfrak{h}_1}_{22}$, $N^{\mathfrak{h}_1}_{23}$, $N^{\mathfrak{h}_1}_{24}$, $N^{\mathfrak{h}_2}_{27}$ and $\Alg_4$ (see Section~\ref{san}). We use the symbols adopted in \cite{BM1} and \cite{BG}.

Most of the relevant Novikov algebras under consideration in dimension two, three and four lead to the construction
of evolution equations from Theorem~\ref{thr} in a triangular form.\footnote{By triangular evolution system we mean such that can be represented in the form: $(u_1)_t = K_1(u_1),\ (u_2)_t = K_2(u_1,u_2), \ (u_3)_t = K_3(u_1,u_2,u_3),\ \ldots\ $.} The only non-triangular systems are associated to the algebras $(N4)$, $(C8)$ and $\Alg_4$.

Contrary to Section~\ref{sec2} we will use here the \lq contravariant\rq~ convention for Novikov algebras (so the formulae agree with the prior work \cite{BM1}).
Therefore the structure constants of a Novikov algebra~$\Alg$, with basis vectors $e_1,\ldots,e_n$, are given by $b^i_{jk}$  such that
\eq{\label{chm}
		(a\cdot b)^i := b^i_{jk}a^j b^k\quad\iff\quad e_i\cdot e_j = b^k_{ij}\,e_k,
}
where $a,b\in\Alg$. The related characteristic matrix is given by $\mathcal{B} = (b_{ij})$, where $b_{ij}:= b^k_{ij}\,e_k$.

\subsection{Dimension one}

The only relevant one-dimensional Novikov algebra $\Alg$ is generated by the rule $u\cdot v := uv$.
This algebra is obviously isomorphic to $\Cm$.

Let $\g_0 = g$, $\g_1=\alpha$, $\h_0 = h$ and $\h_1 = \beta$.  There are no cocycles of second order, thus $\f_1=\f_2=0$. For $c=1$ the evolution equation \eqref{eq} becomes
\eq{\label{KM}
	 g v_t+h v_{xxt} =  \alpha  v_x+3 g v v_x +2 h v_x v_{xx} +h v v_{xxx} +\beta  v_{xxx}\,,
}
here $v\in\Liea_\Cm$. The first order linear term involving the constant $\alpha$ can always be eliminated by a linear change of independent coordinates. The evolution equation \eqref{KM} (with $\alpha=0$) was obtained in \cite{KM}.
In particular for $g=1$ and $h=\alpha=0$ we obtain the Korteweg--de Vries (KdV) equation
\eqq{
	v_{t} = 3 vv_x + \beta v_{xxx}\,,
}
for $g=1$ and $h=-1$ we have the Camassa-Holm equation \cite{CH}
\eqq{
	v_{t} - v_{xxt} = \alpha  v_x + 3 vv_x - 2v_x v_{xx} - v v_{xxx}
				 + \beta v_{xxx}\,.
}
Finally, if $g=\beta=\alpha=0$ and $h=1$ then \eqref{KM} becomes the Hunter-Saxton equation \cite{HS}
\eqq{
	v_{xxt} = 2 v_x v_{xx} + v v_{xxx}\,.
}

\subsection{Two-dimensional algebra $(N4)$}

This Novikov algebra is non-abelian and associative, its structure matrix is
\eqq{
B= \pmatrx{0 & e_1 \\ 0 & e_2}.
}

Consider the most general form of equation \eqref{eq} with the right unity $c=\bra{0,\,  1}^\T$. Thus we take the following bilinear forms generating associated cocycles (see Table~\ref{tab1})
\eqq{
\g_0 = \pmatrx{g_{1 1} & g_{12} \\ g_{12} & g_{22}},
\quad
\f_0=\pmatrx{0 & f \\ -f & 0},\quad
\h_0=\pmatrx{0 & 0 \\ 0 & h},
}
and
\eqq{
\g_1 = \pmatrx{\alpha _{1 1} & \alpha _{12} \\ \alpha _{12} & \alpha _{22}},
\quad
\f_1 = \pmatrx{0 & \gamma \\ -\gamma & 0},
\quad
\h_1 = \pmatrx{0 & 0 \\ 0 & \beta}.
}
Let $\bra{u,\,  v}^\T\in\Liea_\Alg$. Then, the equation \eqref{eq} has the form
\eq{\label{n2}
\begin{split}
g_{11} u_t+g_{12} v_t+f v_{xt} &=  \bra{\alpha _{11} u+ g_{11} uv+\alpha _{12} v + g_{12} v^2+f v v_{x}+\gamma  v_{x}}_x\,,\\
g_{12} u_t+g_{22} v_t-f u_{xt}+h v_{x xt} &=  \Bigl ( \alpha _{12} u + \frac{1}{2}g_{11} u^2 + 2g_{12} uv+\alpha _{22} v +\frac{3}{2} g_{22} v^2 -f u_x v\\
&\qquad -\gamma  u_{x} + \frac{1}{2} h v_x^{\,2} + h v v_{xx} +\beta  v_{xx} \Bigr )_x.
\end{split}
}
\noindent We now show how various examples of 2-component Camassa-Holm equations that have appeared already in the literature fall into this scheme
by identifying the underlying Novikov algebras and bilinear forms.

\begin{example}
For the Novikov algebra $(N4)$ with the particular choice
\eqq{
\g_0 = \pmatrx{1 & 0 \\ 0 & 1},\qquad
\g_1 = \pmatrx{0 & 0 \\ 0 & 0}
}
one obtains the system
\eqq{
 u_t+f v_{xt} &=  \bra{uv +f vv_x+\gamma  v_{x}}_x\,,\\
 v_t-f u_{xt}+h v_{x xt} &= \bra{\frac{1}{2}u^2 + \frac{3}{2} v^2 - f u_x v -\gamma  u_{x}  + \frac{1}{2} h v_x^{\,2}+h v v_{xx}+\beta  v_{xx}}_x\,,
}
which, for $f=\gamma=0\,,$ reduces to
\eqq{
 u_t &=  \bra{uv }_x\,,\\
 v_t+h v_{x xt} &= \bra{\frac{1}{2}u^2 + \frac{3}{2} v^2  + \frac{1}{2} h v_x^{\,2}+h v v_{xx}+\beta  v_{xx}}_x\,.
}
This system, when $\beta=0$, was obtained in \cite{GO} and it is an extension of the Ito equation
\cite{I} ($h=0$).
\end{example}

\begin{example} For the Novikov algebra $(N4)$ with the particular choice
\eqq{
\g_0 = \pmatrx{0 & 1 \\ 1 & 0},\qquad
\g_1 = \pmatrx{\alpha_1 & 0 \\ 0 & \alpha_2}
}
one obtains the system
\eqq{
 v_t+f v_{xt} &=  \bra{\alpha _1 u + v^2 + f vv_x+ \gamma  v_{x}}_x\,,\\
 u_t-f u_{xt}+h v_{x xt} &= \bra{2 uv+\alpha _2 v - f u_x v -\gamma  u_{x} + \frac{1}{2} h v_x^{\,2} +h v v_{xx}+\beta  v_{xx}}_x,
}
and in the case $f=h=0$ this reduces to
\eqq{
 v_t &=  \bra{\alpha _1 u + v^2 + \gamma  v_{x}}_x\,,\\
 u_t &= \bra{2 uv+\alpha _2 v -\gamma  u_{x} +\beta  v_{xx}}_x\,.
}
For $\alpha_2=0$ this system is the dispersive water waves (DWW) equation considered in \cite{Kup}.
If $\beta=\alpha_2=0$ and $\gamma=1$ this system is equivalent to the Kaup-Broer system \cite{K,B}, which under the constraint $u=0$ reduces to the Burgers equation.
\end{example}

\begin{example}
For the Novikov algebra $(N4)$ with the choice
\eqq{
\g_0 = \pmatrx{-1 & 0 \\ 0 & 1},\qquad
\g_1 = \pmatrx{\alpha_1 & 0 \\ 0& \alpha_2}
}
we get the system
\eqq{
u_t-f v_{xt} &=  \bra{uv -f v v_{x}- \alpha_1  u-\gamma  v_{x}}_x\,,\\
v_t-f u_{xt}+h v_{xxt} &=  \bra{-\frac{1}{2}u^2 + \frac{3}{2} v^2+ \alpha_2  v  - f u_x v - \gamma  u_{x}
+\frac{1}{2} h v_x^{\,2}+h v v_{xx}+\beta  v_{xx}}_x\,,
}
which for $h=-1$ and $f=\beta=\alpha_1=0$ is the $2$-component Camassa-Holm equation (CH2)
derived in \cite{LZ,CLZ}, see also \cite{Fa}. If $f=h=0$ the above system reduces to another water wave type equation
\eqq{
 u_t &=  \bra{uv  -\alpha_1  u-\gamma  v_{x}}_x\,,\\
 v_t &= \bra{ -\frac{1}{2}u^2 + \frac{3}{2} v^2  + \alpha_2  v - \gamma  u_{x} +\beta  v_{xx}}_x.
}
\end{example}

\begin{example}
The algebra $(N4)$ is the only Novikov algebra of dimension less or equal to $3$ with non-degenerate
bilinear form generating second order central extension. Thus with the particular choice
\eqq{
\g_0 = \pmatrx{0 & 0 \\ 0 & 0},\quad \f_0=\pmatrx{0 & 1 \\ -1 & 0},
}
(allowed since $\f_0$ is non-degenerate) one obtains the system
\eqq{
 v_{xt} &=  \bra{\alpha _{11} u + \alpha _{12} v + vv_x + \gamma  v_{x}}_x\,,\\
 -u_{xt}+h v_{xxt} &= \bra{\alpha _{12} u + \alpha _{22} v - u_x v -\gamma  u_{x} + \frac{1}{2} h v_x^{\,2} + h v v_{xx}
 +\beta  v_{xx}}_x.
}
\end{example}

\subsection{Three-dimensional Novikov algebra $(C8)$}\label{exx}
This Novikov algebra is a straightforward three-dimensional generalisation of the two-dimensional algebra $(N4)$.
In fact, this algebra admits a straightforward generalization to arbitrary dimensions (see Section \ref{san}).

Here we will consider only one particular case with
\eqq{
\g_0 =\pmatrx{1 & 0 & 0 \\ 0 & 1 & 0 \\ 0 & 0 & 1},
\quad
\f_0 = \pmatrx{ 0 & 0 & f_{1} \\ 0 & 0 & f_{2} \\ -f_{1} & -f_{2} & 0},
\quad
\h_0 = \pmatrx{ 0 & 0 & 0 \\  0 & 0 & 0 \\ 0 & 0 & h}
}
and
\eqq{
\g_1 = \pmatrx{0 & \alpha _1 & 0 \\ \alpha _1 & 0 & 0 \\ 0 & 0 & \alpha _2},
\quad
\f_1 = \pmatrx{ 0 & 0 & \gamma _{1} \\ 0 & 0 & \gamma _{2} \\ -\gamma _{1} & -\gamma _{2} & 0},
\quad
\h_1 = \pmatrx{0 & 0 & 0 \\ 0 & 0 & 0 \\ 0 & 0 & \beta}.
}
Then, for $(u,v,w)^\T\in\Liea_\Alg$ and $c=(0,0,1)^\T$ the system \eqref{eq} takes the form
\eqq{
 u_t+f_1 w_{xt}  &= \bra{\alpha _1 v + uw  + f_1 ww_x +\gamma _1 w_{x}}_x\,,\\
 v_t+f_2 w_{xt} &= \bra{\alpha _1 u +  vw + f_2 ww_x + \gamma _2 w_{x}}_x\,,\\
 w_t-f_1 u_{xt}-f_2 v_{xt}+h w_{xxt} &= \biggr (\frac{1}{2}u^2 + \frac{1}{2}v^2  + \frac{3}{2}w^2+ \alpha _2 w -f_1 u_x w - f_2 v_x w - \gamma _1 u_{x}-\gamma _2 v_{x}\\
 &\qquad  + \frac{1}{2} h w_x^{\,2}+h w w_{xx} +\beta  w_{xx}\biggl )_x\,.
}

For $\alpha_{1,2}= f_{1,2} = \gamma_{1,2} =0$ we obtain:
\eqq{
 u_t  &= \bra{uw}_x\,,\\
 v_t &= \bra{vw}_x\,,\\
 w_t+h w_{xxt} &= \biggr (\frac{1}{2}u^2 + \frac{1}{2}v^2  + \frac{3}{2}w^2  + \frac{1}{2} h w_x^{\,2}+h w w_{xx} +\beta  w_{xx}\biggl )_x\,.
}

\begin{remark}\label{rem2}
This system for $\beta=0$, after the change of dependent variables: $q_1= w + h w_{xx}$, $q_2 = u^2$ and $q_3 = v^2$, is equivalent to the system obtained from CH$(3,1)$  \cite{HI} by the scaling $q_3\map \mu q_3$ followed by the limit $\mu\map 0$.\footnote{We would like to thank the referee for this observation.}
\end{remark}

\section{An $n$-dimensional abelian and associative Novikov algebra $\Tm_n$}\label{saa}

It turns out that many Novikov algebras with nontrivial algebraic properties result in systems of evolution equations
which are degenerate, for example, not fully nonlinear in all of the variables or, as in the example below, triangular. Consider an $n$-dimensional abelian $\Tm_n$ algebra defined by the multiplication rule
\eqq{
	(a\cdot b)^i = \sum_{k=1}^i a^k\, b^{i-k+1}\quad\iff\quad e_i\cdot e_j = \delta_{i+j-1}^k\,e_k\,,
}
where $a,b\in\Tm_n$. The related structure constants are $b^k_{ij} = \delta_{i+j-1}^k$.
For dimensions $n=1,2$ and $4$ the algebra $\Tm_n$ coincides
with Novikov algebras $\Cm$, $(N3)$ and $\widetilde{A}_{3,3}$, respectively.

\begin{proposition}
For any dimension $n$ the algebra $\Tm_n$ is an associative Novikov algebra. Moreover:
\begin{itemize}
\item an arbitrary symmetric bilinear form $\g = (g_{ij})$ on $\Tm_n$, satisfies the quasi-Frobenius condition \eqref{cond1}
if and only if
\eqq{
		g_{ij} = \begin{cases}
g_{1,i+j-1} & \text{for}\quad 1\les i+j-1\les n\\0 & \text{otherwise}
\end{cases},
}
such bilinear form is non-degenerate if $g_{1n}\neq 0\,$;
\item there is no nonzero anti-symmetric bilinear form $\f = (f_{ij})$ on $\Tm_n$ that satisfies the conditions \eqref{cond3} and \eqref{cond4}.
\end{itemize}
\end{proposition}

\noindent The proof is straightforward. Since the algebra $\Tm_n$ is abelian the only condition for bilinear forms generating the central extensions of the Lie algebra $\Liea_{\Tm_n}$ is the quasi-Frobenius condition \eqref{cond1}.

As consequence of the above proposition the most general bilinear forms generating the Poisson pair \eqref{pair} are given by
\eqq{
	(\g_0)_{ij} = g_{i+j-1},\quad (\h_0)_{ij} = h_{i+j-1},\quad (\g_1)_{ij} = \alpha_{i+j-1},\quad (\h_1)_{ij} = \beta_{i+j-1}
}
for $1\les i+j-1\les n$ and $0$ otherwise. Moreover $(\f_0)_{ij} = (\f_1)_{ij} = 0$.

%\newpage

The element $c=(\delta^i_1)\in\Tm_n$ is a unity. For $v=(v^i)\in\Liea_{\Tm_n}$ the equation \eqref{eq} has the following system
\eq{\label{trsl1}
\begin{split}
	g_{i+j-1}v^j_t + h_{i+j-1}v^j_{xxt} &=
	\biggl (\sum_{k=1}^{j}\Bigl (\frac{3}{2}g_{i+j-1}v^kv^{j-k+1} + \frac{1}{2} h_{i+j-1}v^k_x\,v^{j-k+1}_x
	+ h_{i+j-1}v^k v^{j-k+1}_{xx}   \Bigr)\\
	&\qquad +\alpha_{i+j-1}v^j + \beta_{i+j-1}v^j_{xx}\biggr )_x\,,
\end{split}
}
where $1\les i\les n$ and the summation for $j$ is over $1\les j\les n-i+1$. The bilinear form $\h$ is
non-degenerate if $h_n\neq 0$. In this case it is allowed to take $\g_0=0$.

Consider the special case given by
\eqq{
	(\g_0)_{ij} = g \delta^n_{i+j-1},\quad (\h_0)_{ij} = h \delta^n_{i+j-1},\quad (\g_1)_{ij} = \alpha \delta^n_{i+j-1},\quad (\h_1)_{ij} = \beta \delta^n_{i+j-1}.
}
Then, the above system takes the form
\eq{\label{trsl2}
\begin{split}
	g v^i_t + h v^i_{xxt} &=
	\biggl (\sum_{k=1}^{i}\Bigl (\frac{3}{2}g v^kv^{i-k+1} + \frac{1}{2} h v^k_x\,v^{i-k+1}_x
	+ h v^k v^{i-k+1}_{xx}   \Bigr) +\alpha v^i + \beta v^i_{xx}\biggr )_x\,,
\end{split}
}
where $1\les i\les n$.
	
The above systems \eqref{trsl1} and \eqref{trsl2} are triangular: with the only one genuinely nonlinear equation for $i=n$ and $i=1$, respectively, the remaining equations are sequentially linear.
In fact all triangular systems associated with low-dimensional algebras posses similar property.

\section{An $n$-dimensional Novikov algebra $\Alg_n$}\label{san}

Consider the following $n$-dimensional algebra $\Alg_n$, this being a generalisation of the Novikov algebras $(N4)$
and $(C8)$, defined by the rule
\eqq{%\label{An}
		(a\cdot b)^i := a^i b^n\quad\iff\quad e_i\cdot e_j = \delta_i^k\delta_j^n\,e_k,
}
so the structure constants are given by $b^k_{ij} = \delta_i^k\delta_j^n$.

\begin{proposition}\label{pbf}
For any dimension $n$ the algebra $\Alg_n$ is an associative Novikov algebra. If $n\me 2$ it is non-abelian. The associated Lie algebra structure on $\Alg_n$ is non-nilpotent. Moreover:
\begin{itemize}
\item an arbitrary symmetric bilinear form $\g$ on $\Alg_n$ satisfies the quasi-Frobenius condition \eqref{cond1};
\item an anti-symmetric bilinear form $\f = (f_{ij})$ on $\Alg_n$ satisfies the conditions \eqref{cond3}
and \eqref{cond4} if and only if
\eqq{
			f_{ij} = 0\quad\text{for}\quad i\neq n\ \text{and}\ j\neq n\,;
}
\item a symmetric bilinear form $\h = (h_{ij})$ on $\Alg_n$ satisfies the conditions \eqref{cond4} if and only if
\eqq{
		h_{ij} = 0\quad\text{for}\quad i\neq n\ \text{or}\ j\neq n\,.
}
\end{itemize}
\end{proposition}

\noindent The proof is a straightforward calculation and will not be presented.

\begin{remark}
Let $\smf{\Si}$ be the space of smooth functions on the circle $\Si$ and
$\Vect{\Si}$ be the Lie algebra of vector fields $u(x)\pr_x$ on $\Si$.
The Lie bracket in $\Vect{\Si}$ is given by the formula
\eqq{
	\brac{u\,\pr_x,v\,\pr_x} = (uv_x - u_x v)\pr_x\,,\qquad x\in\Si.
}
Consider the semidirect product $\mathscr{G}_n(\Si):= \Vect{\Si}\ltimes \smf{\Si}^{\oplus n}$
of $\Vect{\Si}$ with $n$ copies of $\smf{\Si}$. The Lie bracket in $\mathscr{G}_n(\Si)$
is given by
\eqq{
	\brac{(u\pr_x, \ff),(v\pr_x, \gg)} = \bra{(uv_x - u_x v)\pr_x, u\, \gg_x - v \ff_x},
}
where $(u\pr_x, \ff),(v\pr_x, \gg)\in\mathscr{G}_n(\Si)$.
The Lie algebra $\Liea_{\Alg_n}$ associated with
the $n$-dimensional Novikov algebra $\Alg_n$ is isomorphic to $\mathscr{G}_{n-1}(\Si)$.
The isomorphism $\mathscr{G}_{n-1}(\Si)\arrow\Liea_{\Alg_n}$ is given by $(u\pr_x, \ff)\map - (\ff, u)$.
The particular algebra $\mathscr{G}_{1}(\Si)$, with related integrable systems, was extensively studied,
see \cite{M, MOR, G, GO} and references therein.
\end{remark}

Proposition~\ref{pbf} provides conditions on the bilinear forms generating the central extension
of the Lie algebra $\Liea_{\Alg_n}$. Consequently the most general bilinear forms generating the Poisson pair \eqref{pair} are given by
\eqq{
	(\g_0)_{ij} = g_{ij},\qquad (\f_0)_{ij} = \delta_j^n f_i - \delta_i^n f_j,\qquad (\h_0)_{ij} = \delta_i^n\delta_j^n h
}
and
\eqq{
	(\g_1)_{ij} = \alpha_{ij},\qquad (\f_1)_{ij} = \delta_j^n \gamma_i - \delta_i^n \gamma_j,\qquad (\h_1)_{ij} = \delta_i^n\delta_j^n \beta.
}
The element $c=(c^i)\in\Alg_n$ is the right unity iff $c^n=0$. Hence it is not unique. For $v=(v^i)\in\Liea_{\Alg_n}$ the equation \eqref{eq} has the following form:
\eq{\label{ms}
\begin{split}
i\neq n:\qquad\qquad  g_{ij}v^j_t + f_iv^n_{xt} &= \bra{g_{ij} v^jv^n + f_i v^nv^n_x + \alpha_{ij}v^j + \gamma_iv^n_{x}}_x\,,\\
i=n:\qquad  g_{nj}v^j_t - f_jv^j_{xt} + hv^n_{xxt} &= \Bigl ( g_{nj} v^jv^n + \frac{1}{2}g_{jk}v^jv^k - f_j v^j_xv^n
+ \frac{1}{2} h (v^n_x)^2\\
&\qquad +h v^nv^n_{xx} +  \alpha_{nj}v^j - \gamma_jv^j + \beta v^n_{xx} \Bigr )_x\, .
\end{split}
}
For $n=2$ this system is equivalent to \eqref{n2}.

\begin{example} Consider the particular choice:\footnote{Here the matrices are of dimension $\{n-2,1,1\}\times \{n-2,1,1\}$.}
\eqq{
\g_0 = \pmatrx{{\rm \bm{Id_{n-2}}} & 0 & 0\\ 0 & 0 & 1\\ 0 & 1 & 0},\qquad
\g_1 = \pmatrx{\bm{\alpha} & 0 & 0\\ 0 & \alpha_{n-1} & 0\\ 0 & 0 & \alpha_n},
}
where ${\rm \bm{Id}_{n-2}} = \diag(1,\ldots,1)$ and $\bm{\alpha} = \diag(\alpha_1,\ldots,\alpha_{n-2})$.
Let $\ff = (f_1,\ldots, f_{n-2})^\T$ and $\bg = (\gamma_1,\ldots, \gamma_{n-2})^\T$.
Then, the system \eqref{ms} for $v=(\q,w,u)^\T\in\Liea_{\Alg_n}$ takes the form
\eqq{
\q_t + u_{xt}\ff &= \bra{u\q + uu_x\ff + \bm{\alpha}\q + u_x\bg}_x\,,\\
u_t + f_{n-1} u_{xt} &= \bra{u^2 + f_{n-1}uu_x + \alpha_{n-1}w + \gamma_{n-1}u_x}_x\,,\\
w_t - \ff^\T\q_{xt} - f_{n-1}w_{xt} + h u_{xxt} &= \Bigl (2uw + \frac{1}{2}\q^\T\q
-u\ff^\T\q_x - f_{n-1}uw_x + \frac{1}{2} h u_x^{\,2}\\
&\qquad +h u u_{xx}+ \alpha_{n}u - \bg^\T\q_x - \gamma_{n-1}w_x + \beta u_{xx}\Bigr )_x\,.
}
For $\ff={\bf 0}$ and $f_{n-1}=h=0$ it reduces to the system
\eqq{
\q_t  &= \bra{u\q  + \bm{\alpha}\q + u_x\bg}_x\,,\\
u_t  &= \bra{u^2  + \alpha_{n-1}w + \gamma_{n-1}u_x}_x\,,\\
w_t  &= \Bigl (2uw + \frac{1}{2}\q^\T\q + \alpha_{n}u - \bg^\T\q_x - \gamma_{n-1}w_x + \beta u_{xx}\Bigr )_x\,,
}
which for $\bm{\alpha}={\bf 0}$, $\bg = 0$, $\alpha_n=\beta=0$ and $\alpha_{n-1} =2$, $\gamma_{n-1} = -1$
is the multicomponent dispersive water wave equation considered in \cite{Kup2}.
\end{example}

\begin{example}\label{ex64}
Finally, consider the case of $\g_0 = \diag(1,\ldots,1)$ and $\g_1=0$. Let $\ff = (f_1,\ldots, f_{n-1})^\T$ and $\bg = (\gamma_1,\ldots, \gamma_{n-1})^\T$. Then, taking $v=(\uf,w)^\T\in\Liea_{\Alg_n}$ we get the following multicomponent system
\eqq{
\uf_t + w_{xt}\ff &= \bra{w\uf + ww_x\ff + w_x\bg}_x\,,\\
w_t - \ff^\T\uf_{xt} + h w_{xxt} &= \bra{\frac{3}{2}w^2 + \frac{1}{2}\uf^\T\uf
-w\ff^\T\uf_x + \frac{1}{2} h w_x^{\,2} +h w w_{xx} - \bg^\T\uf_x  + \beta w_{xx}}_x\,,
}
which for $\ff=\bg=0$  reduces to
\eqq{
\uf_t  &= \bra{w\uf}_x\,,\\
w_t + h w_{xxt}  &= \bra{\frac{3}{2}w^2 + \frac{1}{2}\uf^\T\uf   + \frac{1}{2} h w_x^{\,2} +h w w_{xx} + \beta w_{xx}}_x\,.
}
Similarly to Section~\ref{exx} and Remark~\ref{rem2}, this system for $\beta=0$, after the change of dependent variables: $q_1= w + h w_{xx}$, $q_i = \sum_{j=1}^{n-i+1}u_j^{\,2}$, where $i=2,\ldots,n$, is equivalent to the system that one can obtain from CH$(n,1)$  \cite{HI} by scalings $q_{n-i+1}\map \mu_i q_{n-i+1}$, for $i=1\,,\ldots\,,n-2$, followed by the limits $\mu_i\map 0$. There is much scope for the investigation of the links between seemingly different systems via such nonlinear changes of variables and scalings.
\end{example}

\section{Equations of Hydrodynamic type on Novikov algebras}

Taking the dispersionless limit of the evolution equations obtained in Theorem~\ref{thr} gives, in the coordinates $u$, the following
equations of hydrodynamic type:
\eq{
\notag	u_{t_1} &= R^*_c u_x\,,\\
\label{hyde}	 u_{t_2} &=  (R^*_{R_c\La^{-1} u} u)_x + L^*_{\La^{-1}u_x}R^*_cu + \g_1 R_c\La^{-1}u_x\,,
}
where the inertia operator is given by $\Lambda \equiv \g_0$. These are the first two equations from the bi-Hamiltonian hierarchy \eqref{uh} generated by means of the following compatible Poisson operators:
\eqq{
	\P_1\gamma = (R^*_\gamma u)_x + L^*_{\gamma_x}u + \g_1 \gamma_x,\qquad \P_0\gamma = \g_0 \gamma_x,
}
where $\gamma\in\Liea_\Alg$. The first three densities of Hamiltonian functionals are
\eqq{
	H_0 &= g_0(c,v),\\
	H_1 &= \frac{1}{2}\,g_0(v,v\cdot c), \\
	H_2 &= \frac{1}{3}\,g_0\bra{v,v\cdot(v\cdot c)} +\frac{1}{6}\,g_0\bra{v\cdot c,v\cdot v} + \frac{1}{2}\,g_1(v,v\cdot c),
}
where $v = \g_0^{-1}u$. These simplify further when the algebra has a right unity $e$ and one takes $c=e\,,$ or if the algebra has a left unity, in which the
algebra is automatically associative.
For all the explicit examples of Novikov algebras and bilinear forms (with $\det g_0\neq 0$) discussed in this paper the associated Haantjes tensor vanishes.
But only
those systems associated to the algebras $(N4)$, $(C8)$ and $\Alg_n$ are hyperbolic and thus are diagonalisable (see Appendix~\ref{Ahyd}).

Let us consider the dispersionless limit of the $n$-component equation \eqref{ms}. Thus, for $u= (u^i)\in\Liea^*_{\Alg_n}$
such that $(u^i)=(g_{ij}v^j)$ we obtain\footnote{We slightly abuse here the convention previously adopted.}
\eq{\label{hyds}
	u^i_t  = \bra{\eta_{nj} u^j u^i + \alpha^{ij}\eta_{jk}u^k}_x\,,\qquad i=1,\ldots,n\ ,
}
where $\eta = g^{-1}$, such that $(\eta_{ij}) = (g^{ij})$, and $(\alpha^{ij})\equiv(\alpha_{ij})$.

\begin{proposition}
The Haantjes tensor \eqref{haant} of the equations of hydrodynamic type \eqref{hyds} is identically zero.
\end{proposition}
\begin{proof}
We have
\eqq{
	A^i_k = \eta_{nj}u^j\delta^i_k + u^i \eta_{nk} + \alpha^{ij}\eta_{jk}\quad\Longrightarrow\quad
	\pd{A^i_k}{u^j} =  \eta_{nj}\delta^i_k + \eta_{nk}\delta^i_j
}
such that \eqref{hyd}. Hence, the Nijenhuis tensor is given by
\eqq{
N^i_{jk} = \alpha^{ir}\eta_{rj}\eta_{nk} + \eta_{nr}\alpha^{rs}\eta_{sj}\delta^i_k +
\eta_{nj}\eta_{nr}u^r\delta^i_k - \{j\leftrightarrow k\}.
}
Then, after straightforward, but slightly tedious, calculations one can show that the Haantjes
tensor~\eqref{haant} vanishes.
\end{proof}

\noindent It turns out that the eigenvalues/characteristic speeds of the system \eqref{hyds} are, when $n=2\,,3$ or $4\,,$
distinct. We conjecture that this is true for arbitrary sized systems.

In fact, as remarked above, all the dispersionless systems constructed {\sl from the explicit Novikov algebras} in this
paper have vanishing Haantjes tensor. This leads us to the following:

\begin{conjecture} The equations of hydrodynamic type \eqref{hyde} constructed from an arbitrary Novikov algebra have vanishing Haantjes tensor.
\end{conjecture}

\noindent In the case of a commutative Novikov algebra (which are automatically associative) the conjecture is true: in this
case the Nijenhuis tensor vanishes automatically. We have not been able to prove this conjecture in general, even for
those Novikov algebras which arise from a derivation $\partial$ on a commutative, associative algebra (i.e. $a\cdot b=a \partial b$).
We end by noting that the {\sl explicit} construction of flat coordinates for the (inverse) flat metric $g^{ij}(u)=c^{ij}_{~~k} u^k + g_0^{ij}$ with $\det (g^{ij}(u))\neq 0\,$
for a non-commutative Novikov algebra, a problem stated in Balinskii and Novikov \cite{BN} nearly 30 years ago, remains unsolved.

\bigskip

\section*{Acknowledgement}

The authors wishes to thank the anonymous referee for providing to our attention the reference~\cite{HI}
and some important comments on the connections between various Camassa-Holm type systems via nonlinear changes of variables.

\bigskip

\appendix

\section{}

Let $(\alg,[\,,\,])$ be a Lie algebra, $\alg^*$ its (regular) dual space and $\dual{\,,\,}:\alg^*\times\alg\arrow\Cm$ the standard duality pairing. The coadjoint  action $\ad^*:\alg\times\alg^*\arrow\alg^*$ is defined through the formula
\eq{\label{coad}
\dual{\ad^*_a u, b} := -\dual{u,\brac{a,b}},
}
where $a,b\in\alg$ and $u\in\alg^*$.

\subsection{Lie-Poisson structure}\label{liepo}

The Lie-Poisson (or Kirillov-Kostant) bracket in the space of
smooth functions $\smf{\alg^*}$ on the dual Lie algebra $\alg^*$ is given by
\eqq{
\pobr{H, F}(u):=  \ddual{u, \brac{dF, dH}}\,,\qquad u\in\alg^*,
}
where $H, F\in\smf{\alg^*}$ and theirs differentials $dF,dG\in\alg$. The differential of $F\in\smf{\alg^*}$ is defined by the relation
\eq{\label{differ}
		\dual{v, dF}:=  \Diff{F[u+\eps v]}{\eps}{\eps=0},
}
where $v\in\alg^*$. The Hamiltonian equation corresponding to a function $H$ on $\alg^*$ with respect
to the Poisson-Lie structure has the form
\eqq{
		u_t = \P dH \equiv -\ad^*_{dH} u\,,\qquad u\in\alg^*,
}
where $\P$ is the Poisson operator such that $\pobr{H, F} =  \ddual{\P dH, dF}$.

\subsection{Central extension}\label{central}

A $2$-cocycle is a bilinear form $\omega:\alg\times\alg\map\Cm$, which
is skew-symmetric:
\eq{\label{ss}
\omega\bra{a,b} = -\omega\bra{b,a},
}
and it satisfies the cyclic-condition:
\eq{\label{cc}
\omega\bra{[a,b],c} + \omega\bra{[c,a],b} + \omega\bra{[b,c],a} = 0,
}
where $a,b,c\in\alg$. Then, the centrally extended algebra $\tilde{\alg} = \alg\oplus\Cm$ is defined via the Lie bracket
\eqq{
	\brac{(a,\alpha), (b,\beta)} = \bra{[a,b],\omega(a,b)},	
}
where $a,b\in\alg$ and $\alpha,\beta\in\Cm$.

Let the $2$-cocycle $\omega$ be defined by means of a $1$-cocycle $\phi:\alg \arrow \alg^*$, that is
\eq{\label{op}
	\omega\bra{a,b} := \dual{\phi(a),b}.
}
Then, the skew-symmetry \eqref{ss} is equivalent to skew-adjoint's of $\phi$, that is $\phi^\dag=-\phi$,
where $\dual{\phi^\dag(a),b} := \dual{\phi(b),a}$. The cyclic-condition \eqref{cc} impose the following relation on $\phi$:
\eqq{
    \phi\bra{[a,b]} = \ad^*_a\phi(b) - \ad^*_b\phi(a),
}
where $a,b\in\alg$.

The coadjoint action on the centrally extended Lie algebra $\tilde{\alg}$ restricted to $\alg$ is given by
\eqq{
	\widetilde{\ad}^*_a u = \ad^*_au - \delta\phi(a),
}
where $\delta\in\Km$ is a (fixed) charge. For our purposes it can be taken as $\delta=1$. Important is fact that more general cocycles can be obtained as linear compositions of respectively one or two-cocycles.

\subsection{Euler equation}\label{Euler} An invertible self-adjoint operator $\La:\alg\arrow\alg^*$ defining
the quadratic Hamiltonian
\eq{\label{qh}
		H = \frac{1}{2}\dual{u,\La^{-1}u}
} is called an inertia operator on $\alg$. The corresponding Hamiltonian equation on $\alg^*$ is
\eqq{
	u_t =  -\ad^*_{\La^{-1}u}u\,,\qquad u\in\alg^*,
}
since $dH = \La^{-1}u$. This equation is the Euler equation on the dual space $\alg^*$
corresponding to geodesic flows on a Lie Group $G$ associated with $\alg$. In fact the inertia
operator defines the metric $(\, ,\,)_e$ at the identity $e\in G$ such that
$(v ,w)_e= \dual{\La v, w}$, where $v,w\in\alg\equiv T_eG$. For more information on the subject we refer the reader to \cite{AK,KW}.

The Euler equation on the centrally extended Lie algebra $\tilde{\alg}$ for the quadratic Hamiltonian \eqref{qh}
takes the  form
\eqq{
	u_t &=  -\ad^*_{\La^{-1}u}u + \delta\phi\bra{\La^{-1}u},\\
	\delta_t &= 0,
}
where $(u,\delta)\in\tilde{\alg}^*$. The reduction of this Euler equation to $\alg^*$ is natural.

\subsection{Diagonalizability of hydrodynamic systems}\label{Ahyd}	

An evolution equation of hydrodynamic type:
\eq{\label{hyd}
		u^i_t = A^i_k(u)u^k_x\,,\qquad i=1,\ldots,n.
}
is diagonalizable if there are coordinates, so-called Riemann invariants, in which it can be presented in the form
\eqq{
	R^i_t = \Lambda_i(R)R^i_x\qquad \text{\it (no summation)}.
}
Then, if the characteristic speeds $\Lambda_i$ are mutually distinct and satisfy the semi-Hamiltonian condition, which is automatic for Hamiltonian hydrodynamic systems with nondegenerate metric, \eqref{hyd} is
integrable by means of the so-called generalised hodograph method \cite{T}.
We say that the system of hydrodynamic type \eqref{hyd} is hyperbolic if the tensor $A=\left(A^i_k(u)\right)$ has $n$ linearly independent eigenvectors. The diagonalizability criterion
for hyperbolic systems \eqref{hyd} is provided by
vanishing of the Haantjes tensor \cite{H}:
\eq{\label{haant}
%	 H^i_{jk} = A^i_rA^r_sN^s_{jk} + N^i_{rs} A^r_jA^s_k - A^i_rN^r_{sk} A^s_j -A^i_rN^r_{js} A^s_k,
H_A(X,Y)=N_A(AX,AY)-AN_A(X,AY)-AN_A(AX,Y)+A^2 N_A(X,Y)
}
which is defined via the Nijenhuis tensor:
\eqq{
%		N^i_{jk} = A^r_j\pd{A^i_k}{u^r} - A^r_k\pd{A^i_j}{u^r} + A^i_r\pd{A^r_j}{u^k} -A^i_r\pd{A^r_k}{u^j}.
N_A(X,Y)=[AX,AY]-A[X,AY]-A[AX,Y]+A^2[X,Y]\,.
}
More precisely, a hyperbolic system of hydrodynamic type \eqref{hyd} is diagonalizable if and only if the corresponding Haantjes tensor \eqref{haant} vanishes.

\section{}\label{class}

Here we present a classification of bilinear forms generating central extensions of the Lie algebras~$\Liea_{\Alg}$ (see Section~\ref{cene}). This classification is based on the classification of Novikov algebras in dimension $\les 4$ presented in \cite{BM1} and \cite{BG}.
We use the same names for the different types of Novikov algebras as used in \cite{BM1} and \cite{BG}.
The most general form of the bilinear forms generating
cocycles of order one, two and three, are denoted by $g$, $f$ and $h$ respectively.
Our classification extends that of \cite{BM1b}, where the symmetric bilinear forms satisfying the quasi-Frobenius condition were presented.
The characteristic matrix of a Novikov algebra $\Alg$ is  $\mathcal{B} = (b_{ij})$ defined by $b_{ij}:= e_i\cdot e_j = b^k_{ij}\,e_k$ according to \eqref{chm}. From the classification we exclude the Novikov algebras that are trivial, i.e. $b_{ij} = 0$,
or can be represented as a direct sum of lower dimensional Novikov algebras. In dimension four we only consider
non-transitive Novikov algebras.
The classification is presented in Tables~\ref{tab1}--\ref{tab6}.\footnote{The special cases $(A6)'$, $(C10)'$ were overlooked in \cite{BM1b}. The case $(D6)'$ must be distinguished because of the form~$f$. In \cite{BG} there is  a misprint in the multiplication table of the algebra $N^{\mathfrak{h}_1}_{20}$.}

\bigskip\bigskip
%\newpage

\arraycolsep=4.5pt
\setlength{\tabcolsep}{3.5pt}

{\footnotesize
\begin{table}[h!]
\caption{Classification of bilinear forms associated with one and two-dimensional Novikov algebras.}\label{tab1}
\begin{tabular}{ c  c  c  c  c  c }
\hline type & charact. matrix & $g$ & $f$ & $h$ & comments \\\hline\\[-10pt]

$\Cm$ & $e_1$& $g_{11}$ & $0$ & $h_{11}$ & \smallskip\\

$(T2)$ & $
\pmatrx{
 e_2 & 0 \\
 0 & 0 \\
}
$ & $
\pmatrx{
 g_{1 1} & g_{1 2} \\
 g_{1 2} & 0 \\
}
$ & $
\pmatrx{
 0 & 0 \\
 0 & 0 \\
}
$ & $
\pmatrx{
 h_{1 1} & h_{1 2} \\
 h_{1 2} & 0 \\
}
$ & transitive\smallskip\\

$(T3)$ & $
\pmatrx{
 0 & 0 \\
 -e_1 & 0 \\
}
$ & $
\pmatrx{
 0 & g_{1 2} \\
 g_{1 2} & g_{2 2} \\
}
$ & $
\pmatrx{
 0 & 0 \\
 0 & 0 \\
}
$ & $
\pmatrx{
 0 & 0 \\
 0 & h_{2 2} \\
}
$ & transitive\smallskip\\

$(N3)$ & $
\pmatrx{
 e_1 & e_2 \\
 e_2 & 0 \\
}
$ & $
\pmatrx{
 g_{1 1} & g_{1 2} \\
 g_{1 2} & 0 \\
}
$ & $
\pmatrx{
 0 & 0 \\
 0 & 0 \\
}
$ & $
\pmatrx{
 h_{1 1} & h_{1 2} \\
 h_{1 2} & 0 \\
}
$ & \smallskip\\

$(N4)$ & $
\pmatrx{
 0 & e_1 \\
 0 & e_2 \\
}
$ & $
\pmatrx{
 g_{1 1} & g_{1 2} \\
 g_{1 2} & g_{2 2} \\
}
$ & $
\pmatrx{
 0 & f_{1 2} \\
 -f_{1 2} & 0 \\
}
$ & $
\pmatrx{
 0 & 0 \\
 0 & h_{2 2} \\
}
$ & $\det f\neq 0$\smallskip\\

$(N5)$ & $
\pmatrx{
 0 & e_1 \\
 0 & e_1+e_2 \\
}
$ & $
\pmatrx{
 0 & g_{1 2} \\
 g_{1 2} & g_{2 2} \\
}
$ & $
\pmatrx{
 0 & 0 \\
 0 & 0 \\
}
$ & $
\pmatrx{
 0 & 0 \\
 0 & h_{2 2} \\
}
$ & \smallskip\\

$(N6)$ &
$
\pmatrx{
 0 & e_1 \\
 \kappa  e_1 & e_2 \\
}
$ & $
\pmatrx{
 0 & g_{1 2} \\
 g_{1 2} & g_{2 2} \\
}
$ & $
\pmatrx{
 0 & 0 \\
 0 & 0 \\
}
$ & $
\pmatrx{
 0 & 0 \\
 0 & h_{2 2} \\
}
$ & $\kappa\neq 0,1$ \smallskip\\ \hline

\end{tabular}
\end{table}
}

%\bigskip
\newpage

{\footnotesize
\begin{table}[h!]
\caption{Classification of bilinear forms associated with three-dimensional Novikov algebras of type $A$.}\label{tab2}
\begin{tabular}{ c  c  c  c  c  c }
\hline type & charact. matrix & $g$ & $f$ & $h$ & comments \\\hline\\[-10pt]

$(A3)$ & $
\pmatrx{
 0 & 0 & 0 \\
 0 & e_1 & 0 \\
 0 & 0 & e_1 \\
}
$ & $
\pmatrx{
 0 & 0 & 0 \\
 0 & g_{2 2} & g_{2 3} \\
 0 & g_{2 3} & g_{3 3} \\
}
$ & $
\pmatrx{
 0 & 0 & 0 \\
 0 & 0 & f_{2 3} \\
 0 & -f_{2 3} & 0 \\
}
$ & $
\pmatrx{
 0 & 0 & 0 \\
 0 & h_{2 2} & h_{2 3} \\
 0 & h_{2 3} & h_{3 3} \\
}
$ & $\begin{array}{c} {\rm transitive}\\ \det g = 0\end{array}$\smallskip\\

$(A4)$ & $
\pmatrx{
 0 & 0 & 0 \\
 0 & 0 & e_1 \\
 0 & e_1 & e_2 \\
}
$ & $
\pmatrx{
 0 & 0 & g_{2 2} \\
 0 & g_{2 2} & g_{2 3} \\
 g_{2 2} & g_{2 3} & g_{3 3} \\
}
$ & $
\pmatrx{
 0 & 0 & 0 \\
 0 & 0 & 0 \\
 0 & 0 & 0 \\
}
$ & $
\pmatrx{
 0 & 0 & h_{2 2} \\
 0 & h_{2 2} & h_{2 3} \\
 h_{2 2} & h_{2 3} & h_{3 3} \\
}
$ & transitive\smallskip\\

$(A5)$ & $
\pmatrx{
 0 & 0 & 0 \\
 0 & 0 & e_1 \\
 0 & -e_1 & 0 \\
}
$ & $
\pmatrx{
 0 & 0 & 0 \\
 0 & g_{2 2} & g_{2 3} \\
 0 & g_{2 3} & g_{3 3} \\
}
$ & $
\pmatrx{
 0 & 0 & 0 \\
 0 & 0 & f_{2 3} \\
 0 & -f_{2 3} & 0 \\
}
$ & $
\pmatrx{
 0 & 0 & 0 \\
 0 & h_{2 2} & h_{2 3} \\
 0 & h_{2 3} & h_{3 3} \\
}
$ & $\begin{array}{c} {\rm transitive}\\ \det g = 0\end{array}$ \smallskip\\

$(A6)$ & $
\pmatrx{
 0 & 0 & 0 \\
 0 & e_1 & e_1 \\
 0 & -e_1 & \kappa  e_1 \\
}
$ & $
\pmatrx{
 0 & 0 & 0 \\
 0 & g_{2 2} & g_{2 3} \\
 0 & g_{2 3} & g_{3 3} \\
}
$ & $
\pmatrx{
 0 & 0 & 0 \\
 0 & 0 & f_{2 3} \\
 0 & -f_{2 3} & 0 \\
}
$ & $
\pmatrx{
 0 & 0 & 0 \\
 0 & h_{2 2} & h_{2 3} \\
 0 & h_{2 3} & h_{3 3} \\
}
$ & $\begin{array}{c}\kappa \neq -1\\ {\rm transitive}\\ \det g = 0\end{array}$ \smallskip\\

$(A6)'$ & $
\pmatrx{
 0 & 0 & 0 \\
 0 & e_1 & e_1 \\
 0 & -e_1 & \kappa e_1 \\
}
$ & $
\pmatrx{
 0 & -g_{1 3} & g_{1 3} \\
 -g_{1 3} & g_{2 2} & g_{2 3} \\
 g_{1 3} & g_{2 3} & g_{3 3} \\
}
$ & $
\pmatrx{
 0 & 0 & 0 \\
 0 & 0 & f_{2 3} \\
 0 & -f_{2 3} & 0 \\
}
$ & $
\pmatrx{
 0 & 0 & 0 \\
 0 & h_{2 2} & h_{2 3} \\
 0 & h_{2 3} & h_{3 3} \\
}
$ & $\begin{array}{c}\kappa =-1\\ {\rm transitive}\end{array}$ \smallskip\\

$(A7)$ & $
\pmatrx{
 0 & 0 & 0 \\
 0 & 0 & e_1 \\
 0 & \kappa  e_1 & e_2 \\
}
$ & $
\pmatrx{
 0 & 0 & g_{2 2} \\
 0 & g_{2 2} & g_{2 3} \\
 g_{2 2} & g_{2 3} & g_{3 3} \\
}
$ & $
\pmatrx{
 0 & 0 & 0 \\
 0 & 0 & 0 \\
 0 & 0 & 0 \\
}
$ & $
\pmatrx{
 0 & 0 & 0 \\
 0 & 0 & h_{2 3} \\
 0 & h_{2 3} & h_{3 3} \\
}
$ & $\begin{array}{c}\kappa \neq 1\\ {\rm transitive}\end{array}$ \smallskip\\

$(A8)$ & $
\pmatrx{
 0 & 0 & 0 \\
 0 & 0 & 0 \\
 0 & e_1 & e_2 \\
}
$ & $
\pmatrx{
 0 & 0 & g_{1 3} \\
 0 & 0 & g_{2 3} \\
 g_{1 3} & g_{2 3} & g_{3 3} \\
}
$ & $
\pmatrx{
 0 & 0 & 0 \\
 0 & 0 & 0 \\
 0 & 0 & 0 \\
}
$ & $
\pmatrx{
 0 & 0 & 0 \\
 0 & 0 & h_{2 3} \\
 0 & h_{2 3} & h_{3 3} \\
}
$ & $\begin{array}{c} {\rm transitive}\\ \det g = 0\end{array}$ \smallskip\\

$(A10)$ & $
\pmatrx{
 0 & 0 & 0 \\
 0 & 0 & 0 \\
 0 & e_2 & e_1 \\
}
$ & $
\pmatrx{
 0 & 0 & g_{1 3} \\
 0 & 0 & g_{2 3} \\
 g_{1 3} & g_{2 3} & g_{3 3} \\
}
$ & $
\pmatrx{
 0 & 0 & 0 \\
 0 & 0 & 0 \\
 0 & 0 & 0 \\
}
$ & $
\pmatrx{
 0 & 0 & h_{1 3} \\
 0 & 0 & 0 \\
 h_{1 3} & 0 & h_{3 3} \\
}
$ & $\begin{array}{c} {\rm transitive}\\ \det g = 0\end{array}$ \smallskip\\

$(A11)$ & $
\pmatrx{
 0 & 0 & 0 \\
 0 & 0 & 0 \\
 e_1 & \kappa  e_2 & 0 \\
}
$ & $
\pmatrx{
 0 & 0 & g_{1 3} \\
 0 & 0 & g_{2 3} \\
 g_{1 3} & g_{2 3} & g_{3 3} \\
}
$ & $
\pmatrx{
 0 & 0 & 0 \\
 0 & 0 & 0 \\
 0 & 0 & 0 \\
}
$ & $
\pmatrx{
 0 & 0 & 0 \\
 0 & 0 & 0 \\
 0 & 0 & h_{3 3} \\
}
$ & $\begin{array}{c}\abs{\kappa}\les 1,\ \kappa \neq 0\\ {\rm transitive}\\ \det g = 0\end{array}$ \smallskip\\

$(A12)$ & $
\pmatrx{
 0 & 0 & 0 \\
 0 & 0 & 0 \\
 e_1 & e_1+e_2 & 0 \\
}
$ & $
\pmatrx{
 0 & 0 & g_{1 3} \\
 0 & 0 & g_{2 3} \\
 g_{1 3} & g_{2 3} & g_{3 3} \\
}
$ & $
\pmatrx{
 0 & 0 & 0 \\
 0 & 0 & 0 \\
 0 & 0 & 0 \\
}
$ & $
\pmatrx{
 0 & 0 & 0 \\
 0 & 0 & 0 \\
 0 & 0 & h_{3 3} \\
}
$ & $\begin{array}{c} {\rm transitive}\\ \det g = 0\end{array}$\smallskip\\

$(A13)$ & $
\pmatrx{
 0 & 0 & 0 \\
 0 & e_1 & 0 \\
 e_1 & \frac{e_2}{2} & 0 \\
}
$ & $
\pmatrx{
 0 & 0 & \frac{1}{2}g_{2 2} \\
 0 & g_{2 2} & g_{2 3} \\
 \frac{1}{2}g_{2 2} & g_{2 3} & g_{3 3} \\
}
$ & $
\pmatrx{
 0 & 0 & 0 \\
 0 & 0 & 0 \\
 0 & 0 & 0 \\
}
$ & $
\pmatrx{
 0 & 0 & 0 \\
 0 & 0 & 0 \\
 0 & 0 & h_{3 3} \\
}
$ & transitive\smallskip\\\hline

\end{tabular}
\end{table}
}

\newpage

{\footnotesize
\begin{table}[h!]
\caption{Classification of bilinear forms associated with three-dimensional Novikov algebras of type $C$.}\label{tab3}
\begin{tabular}{ c  c  c  c  c  c }
\hline type & charact. matrix & $g$ & $f$ & $h$ & comments \\\hline\\[-10pt]

$(C6)$ & $
\pmatrx{
 0 & 0 & e_1 \\
 0 & 0 & e_2 \\
 e_1 & 0 & e_3 \\
}
$ & $
\pmatrx{
 0 & 0 & g_{1 3} \\
 0 & g_{2 2} & g_{2 3} \\
 g_{1 3} & g_{2 3} & g_{3 3} \\
}
$ & $
\pmatrx{
 0 & 0 & 0 \\
 0 & 0 & f_{2 3} \\
 0 & -f_{2 3} & 0 \\
}
$ & $
\pmatrx{
 0 & 0 & h_{1 3} \\
 0 & 0 & 0 \\
 h_{1 3} & 0 & h_{3 3} \\
}
$ & \smallskip\\

$(C7)$ & $
\pmatrx{
 0 & 0 & e_1 \\
 0 & 0 & e_2 \\
 e_1 & 0 & e_2+e_3 \\
}
$ & $
\pmatrx{
 0 & 0 & g_{1 3} \\
 0 & 0 & g_{2 3} \\
 g_{1 3} & g_{2 3} & g_{3 3} \\
}
$ & $
\pmatrx{
 0 & 0 & 0 \\
 0 & 0 & 0 \\
 0 & 0 & 0 \\
}
$ & $
\pmatrx{
 0 & 0 & h_{1 3} \\
 0 & 0 & 0 \\
 h_{1 3} & 0 & h_{3 3} \\
}
$ & $\det g = 0$ \smallskip\\

$(C8)$ & $
\pmatrx{
 0 & 0 & e_1 \\
 0 & 0 & e_2 \\
 0 & 0 & e_3 \\
}
$ & $
\pmatrx{
 g_{1 1} & g_{1 2} & g_{1 3} \\
 g_{1 2} & g_{2 2} & g_{2 3} \\
 g_{1 3} & g_{2 3} & g_{3 3} \\
}
$ & $
\pmatrx{
 0 & 0 & f_{1 3} \\
 0 & 0 & f_{2 3} \\
 -f_{1 3} & -f_{2 3} & 0 \\
}
$ & $
\pmatrx{
 0 & 0 & 0 \\
 0 & 0 & 0 \\
 0 & 0 & h_{3 3} \\
}
$ & \smallskip\\

$(C9)$ & $
\pmatrx{
 0 & 0 & e_1 \\
 0 & 0 & e_2 \\
 \kappa  e_1 & 0 & e_3 \\
}
$ & $
\pmatrx{
 0 & 0 & g_{1 3} \\
 0 & g_{2 2} & g_{2 3} \\
 g_{1 3} & g_{2 3} & g_{3 3} \\
}
$ & $
\pmatrx{
 0 & 0 & 0 \\
 0 & 0 & f_{2 3} \\
 0 & -f_{2 3} & 0 \\
}
$ & $
\pmatrx{
 0 & 0 & 0 \\
 0 & 0 & 0 \\
 0 & 0 & h_{3 3} \\
}
$ & $\kappa \neq 0,1$ \smallskip\\

$(C10)$ & $
\pmatrx{
 0 & 0 & e_1 \\
 0 & 0 & e_2 \\
 \kappa  e_1 & 0 & e_2+e_3 \\
}
$ & $
\pmatrx{
 0 & 0 & g_{1 3} \\
 0 & 0 & g_{2 3} \\
 g_{1 3} & g_{2 3} & g_{3 3} \\
}
$ & $
\pmatrx{
 0 & 0 & 0 \\
 0 & 0 & 0 \\
 0 & 0 & 0 \\
}
$ & $
\pmatrx{
 0 & 0 & 0 \\
 0 & 0 & 0 \\
 0 & 0 & h_{3 3} \\
}
$ & $\begin{array}{c}\kappa \neq 0,1\\ \det g = 0\end{array}$ \smallskip\\

$(C10)'$ & $
\pmatrx{
 0 & 0 & e_1 \\
 0 & 0 & e_2 \\
 \kappa  e_1 & 0 & e_2+e_3 \\
}
$ & $
\pmatrx{
 g_{1 1} & 0 & g_{1 3} \\
 0 & 0 & g_{2 3} \\
 g_{1 3} & g_{2 3} & g_{3 3} \\
}
$ & $
\pmatrx{
 0 & 0 & f_{1 3} \\
 0 & 0 & 0 \\
 -f_{1 3} & 0 & 0 \\
}
$  & $
\pmatrx{
 0 & 0 & 0 \\
 0 & 0 & 0 \\
 0 & 0 & h_{3 3} \\
}
$ & $\kappa = 0$ \smallskip\\

$(C11)$ & $
\pmatrx{
 0 & 0 & e_1 \\
 0 & 0 & e_2 \\
 e_1 & e_2 & e_3 \\
}
$ & $
\pmatrx{
 0 & 0 & g_{1 3} \\
 0 & 0 & g_{2 3} \\
 g_{1 3} & g_{2 3} & g_{3 3} \\
}
$ & $
\pmatrx{
 0 & 0 & 0 \\
 0 & 0 & 0 \\
 0 & 0 & 0 \\
}
$ & $
\pmatrx{
 0 & 0 & h_{1 3} \\
 0 & 0 & h_{2 3} \\
 h_{1 3} & h_{2 3} & h_{3 3} \\
}
$ & $\det g=0$\smallskip\\

$(C12)$ & $
\pmatrx{
 0 & 0 & e_1 \\
 0 & 0 & e_2 \\
 e_1 & \kappa  e_2 & e_3 \\
}
$ & $
\pmatrx{
 0 & 0 & g_{1 3} \\
 0 & 0 & g_{2 3} \\
 g_{1 3} & g_{2 3} & g_{3 3} \\
}
$ & $
\pmatrx{
 0 & 0 & 0 \\
 0 & 0 & 0 \\
 0 & 0 & 0 \\
}
$ & $
\pmatrx{
 0 & 0 & h_{1 3} \\
 0 & 0 & 0 \\
 h_{1 3} & 0 & h_{3 3} \\
}
$ & $\begin{array}{c}\kappa \neq 0,1\\ \det g=0\end{array}$ \smallskip\\

$(C13)$ & $
\pmatrx{
 0 & 0 & e_1 \\
 0 & 0 & e_2 \\
 \gamma  e_1 & \kappa  e_2 & e_3 \\
}
$ & $
\pmatrx{
 0 & 0 & g_{1 3} \\
 0 & 0 & g_{2 3} \\
 g_{1 3} & g_{2 3} & g_{3 3} \\
}
$ & $
\pmatrx{
 0 & 0 & 0 \\
 0 & 0 & 0 \\
 0 & 0 & 0 \\
}
$ & $
\pmatrx{
 0 & 0 & 0 \\
 0 & 0 & 0 \\
 0 & 0 & h_{3 3} \\
}
$ & $\begin{array}{c}\gamma ,\kappa \neq 0,1\\ \det g=0\end{array}$ \smallskip\\

$(C14)$ & $
\pmatrx{
 0 & 0 & e_1 \\
 0 & 0 & e_2 \\
 e_1 & e_1+e_2 & e_3 \\
}
$ & $
\pmatrx{
 0 & 0 & g_{1 3} \\
 0 & 0 & g_{2 3} \\
 g_{1 3} & g_{2 3} & g_{3 3} \\
}
$ & $
\pmatrx{
 0 & 0 & 0 \\
 0 & 0 & 0 \\
 0 & 0 & 0 \\
}
$ & $
\pmatrx{
 0 & 0 & 0 \\
 0 & 0 & h_{2 3} \\
 0 & h_{2 3} & h_{3 3} \\
}
$ & $\det g=0$ \smallskip\\

$(C15)$ & $
\pmatrx{
 0 & 0 & e_1 \\
 0 & 0 & e_2 \\
 \kappa  e_1 & e_1+\kappa  e_2 & e_3 \\
}
$ & $
\pmatrx{
 0 & 0 & g_{1 3} \\
 0 & 0 & g_{2 3} \\
 g_{1 3} & g_{2 3} & g_{3 3} \\
}
$ & $
\pmatrx{
 0 & 0 & 0 \\
 0 & 0 & 0 \\
 0 & 0 & 0 \\
}
$ & $
\pmatrx{
 0 & 0 & 0 \\
 0 & 0 & 0 \\
 0 & 0 & h_{3 3} \\
}
$ & $\begin{array}{c}\kappa \neq 0,1\\ \det g=0\end{array}$ \smallskip\\

$(C16)$ & $
\pmatrx{
 0 & 0 & e_1 \\
 0 & 0 & e_2 \\
 0 & e_1 & e_3 \\
}
$ & $
\pmatrx{
 0 & 0 & g_{1 3} \\
 0 & g_{2 2} & g_{2 3} \\
 g_{1 3} & g_{2 3} & g_{3 3} \\
}
$ & $
\pmatrx{
 0 & 0 & 0 \\
 0 & 0 & f_{2 3} \\
 0 & -f_{2 3} & 0 \\
}
$ & $
\pmatrx{
 0 & 0 & 0 \\
 0 & 0 & 0 \\
 0 & 0 & h_{3 3} \\
}
$ & \smallskip\\

$(C17)$ & $
\pmatrx{
 0 & 0 & e_1 \\
 0 & 0 & e_2 \\
 0 & e_1 & e_2+e_3 \\
}
$ & $
\pmatrx{
 0 & 0 & g_{1 3} \\
 0 & 0 & g_{2 3} \\
 g_{1 3} & g_{2 3} & g_{3 3} \\
}
$ & $
\pmatrx{
 0 & 0 & 0 \\
 0 & 0 & 0 \\
 0 & 0 & 0 \\
}
$ & $
\pmatrx{
 0 & 0 & 0 \\
 0 & 0 & 0 \\
 0 & 0 & h_{3 3} \\
}
$ & $\det g=0$\smallskip\\

$(C18)$ & $
\pmatrx{
 0 & 0 & e_1+e_2 \\
 0 & 0 & e_2 \\
 0 & -e_2 & e_3 \\
}
$ & $
\pmatrx{
 g_{1 1} & 0 & g_{1 3} \\
 0 & 0 & 0 \\
 g_{1 3} & 0 & g_{3 3} \\
}
$ & $
\pmatrx{
 0 & 0 & f_{1 3} \\
 0 & 0 & 0 \\
 -f_{1 3} & 0 & 0 \\
}
$ & $
\pmatrx{
 0 & 0 & 0 \\
 0 & 0 & 0 \\
 0 & 0 & h_{3 3} \\
}
$ & $\det g=0$\smallskip\\

$(C19)$ & $
\pmatrx{
 0 & 0 & e_1+e_2 \\
 0 & 0 & e_2 \\
 0 & -e_2 & e_1+e_3 \\
}
$ & $
\pmatrx{
 g_{2 3} & 0 & g_{1 3} \\
 0 & 0 & g_{2 3} \\
 g_{1 3} & g_{2 3} & g_{3 3} \\
}
$ & $
\pmatrx{
 0 & 0 & 0 \\
 0 & 0 & 0 \\
 0 & 0 & 0 \\
}
$ & $
\pmatrx{
 0 & 0 & 0 \\
 0 & 0 & 0 \\
 0 & 0 & h_{3 3} \\
}
$ & \smallskip\\\hline

\end{tabular}
\end{table}
}

%\bigskip

\newpage

{\footnotesize
\begin{table}[h!]
\caption{Classification of bilinear forms associated with three-dimensional Novikov algebras of type $D$.}\label{tab4}
\begin{tabular}{ c  c  c  c  c  c }
\hline type & charact. matrix & $g$ & $f$ & $h$ & comments \\\hline\\[-10pt]

$(D2)$ & $
\pmatrx{
 e_2 & 0 & e_1 \\
 0 & 0 & e_2 \\
 e_1 & e_2 & e_3 \\
}
$ & $
\pmatrx{
 g_{11} & 0 & g_{1 3} \\
 0 & 0 & g_{11} \\
 g_{1 3} & g_{11} & g_{3 3} \\
}
$ & $
\pmatrx{
 0 & 0 & 0 \\
 0 & 0 & 0 \\
 0 & 0 & 0 \\
}
$ & $
\pmatrx{
 h_{11} & 0 & h_{1 3} \\
 0 & 0 & h_{11} \\
 h_{1 3} & h_{11} & h_{3 3} \\
}
$ & $\det h\neq 0$\smallskip\\

$(D3)$ & $
\pmatrx{
 e_2 & 0 & e_1 \\
 0 & 0 & e_2 \\
 e_1+e_2 & e_2 & e_3 \\
}
$ & $
\pmatrx{
 g_{2 3} & 0 & g_{1 3} \\
 0 & 0 & g_{2 3} \\
 g_{1 3} & g_{2 3} & g_{3 3} \\
}
$ & $
\pmatrx{
 0 & 0 & 0 \\
 0 & 0 & 0 \\
 0 & 0 & 0 \\
}
$ & $
\pmatrx{
 0 & 0 & h_{1 3} \\
 0 & 0 & 0 \\
 h_{1 3} & 0 & h_{3 3} \\
}
$ & \smallskip\\

$(D4)$ & $
\pmatrx{
 e_2 & 0 & e_1 \\
 0 & 0 & e_2 \\
 \frac{e_1}{2} & 0 & e_3 \\
}
$ & $
\pmatrx{
 2 g_{2 3} & 0 & g_{1 3} \\
 0 & 0 & g_{2 3} \\
 g_{1 3} & g_{2 3} & g_{3 3} \\
}
$ & $
\pmatrx{
 0 & 0 & 0 \\
 0 & 0 & 0 \\
 0 & 0 & 0 \\
}
$ & $
\pmatrx{
 0 & 0 & 0 \\
 0 & 0 & 0 \\
 0 & 0 & h_{3 3} \\
}
$ & \smallskip\\

$(D5)$ & $
\pmatrx{
 e_2 & 0 & e_1 \\
 0 & 0 & e_2 \\
 \frac{e_1}{2} & 0 & e_2+e_3 \\
}
$ & $
\pmatrx{
 2 g_{2 3} & 0 & g_{1 3} \\
 0 & 0 & g_{2 3} \\
 g_{1 3} & g_{2 3} & g_{3 3} \\
}
$ & $
\pmatrx{
 0 & 0 & 0 \\
 0 & 0 & 0 \\
 0 & 0 & 0 \\
}
$ & $
\pmatrx{
 0 & 0 & 0 \\
 0 & 0 & 0 \\
 0 & 0 & h_{3 3} \\
}
$ & \smallskip\\

$(D6)$ & $
\pmatrx{
 e_2 & 0 & e_1 \\
 0 & 0 & e_2 \\
 \kappa  e_1 & (2 \kappa -1) e_2 & e_3 \\
}
$ & $
\pmatrx{
 g_{1 1} & 0 & g_{1 3} \\
 0 & 0 & \kappa  g_{1 1} \\
 g_{1 3} & \kappa  g_{1 1} & g_{3 3} \\
}
$ & $
\pmatrx{
 0 & 0 & 0 \\
 0 & 0 & 0 \\
 0 & 0 & 0 \\
}
$ & $
\pmatrx{
 0 & 0 & 0 \\
 0 & 0 & 0 \\
 0 & 0 & h_{3 3} \\
}
$ & $\kappa \neq 0,\frac{1}{2},1$ \smallskip\\

$(D6)'$ & $
\pmatrx{
 e_2 & 0 & e_1 \\
 0 & 0 & e_2 \\
 \kappa  e_1 & (2 \kappa -1) e_2 & e_3 \\
}
$ & $
\pmatrx{
 g_{1 1} & 0 & g_{1 3} \\
 0 & 0 & 0 \\
 g_{1 3} & 0 & g_{3 3} \\
}
$ & $
\pmatrx{
 0 & 0 & f_{1 3} \\
 0 & 0 & 0 \\
 -f_{1 3} & 0 & 0 \\
}
$ & $
\pmatrx{
 0 & 0 & 0 \\
 0 & 0 & 0 \\
 0 & 0 & h_{3 3} \\
}
$ & $\begin{array}{c}\kappa=0\\ \det g=0\end{array}$\smallskip\\\hline

\end{tabular}
\end{table}
}

\bigskip\bigskip

{\footnotesize
\begin{table}[h!]
\caption{Classification of bilinear forms associated with all four-dimensional abelian nontransitive Novikov algebras and additionally the Novikov algebra $\Alg_4$.}\label{tab5}
\begin{tabular}{ c  c  c  c  c  c }
\hline type & charact. matrix & $g$ & $f$ & $h$ & comments \\\hline\\[-10pt]

$\widetilde{A}_{3,1}$ & $
\pmatrx{
 e_1 & e_2 & e_3 & e_4 \\
 e_2 & 0 & 0 & 0 \\
 e_3 & 0 & 0 & 0 \\
 e_4 & 0 & 0 & 0 \\
}
$ & $
\pmatrx{
 g_{11} & g_{12} & g_{13} & g_{14} \\
 g_{12} & 0 & 0 & 0 \\
 g_{13} & 0 & 0 & 0 \\
 g_{14} & 0 & 0 & 0 \\
}
$ & $
\pmatrx{
 0 & 0 & 0 & 0 \\
 0 & 0 & 0 & 0 \\
 0 & 0 & 0 & 0 \\
 0 & 0 & 0 & 0 \\
}
$ & $
\pmatrx{
 h_{11} & h_{12} & h_{13} & h_{14} \\
 h_{12} & 0 & 0 & 0 \\
 h_{13} & 0 & 0 & 0 \\
 h_{14} & 0 & 0 & 0 \\
}
$ & $\det g =0$\smallskip\\

$\widetilde{A}_{3,2}$ & $
\pmatrx{
 e_1 & e_2 & e_3 & e_4 \\
 e_2 & e_3 & 0 & 0 \\
 e_3 & 0 & 0 & 0 \\
 e_4 & 0 & 0 & 0 \\
}
$ & $
\pmatrx{
 g_{11} & g_{12} & g_{22} & g_{14} \\
 g_{12} & g_{22} & 0 & 0 \\
 g_{22} & 0 & 0 & 0 \\
 g_{14} & 0 & 0 & 0 \\
}
$ & $
\pmatrx{
 0 & 0 & 0 & 0 \\
 0 & 0 & 0 & 0 \\
 0 & 0 & 0 & 0 \\
 0 & 0 & 0 & 0 \\
}
$ & $
\pmatrx{
 h_{11} & h_{12} & h_{22} & h_{14} \\
 h_{12} & h_{22} & 0 & 0 \\
 h_{22} & 0 & 0 & 0 \\
 h_{14} & 0 & 0 & 0 \\
}
$ & $\det g =0$\smallskip\\

$\widetilde{A}_{3,3}$ & $
\pmatrx{
 e_1 & e_2 & e_3 & e_4 \\
 e_2 & e_3 & e_4 & 0 \\
 e_3 & e_4 & 0 & 0 \\
 e_4 & 0 & 0 & 0 \\
}
$ & $
\pmatrx{
 g_{11} & g_{12} & g_{22} & g_{23} \\
 g_{12} & g_{22} & g_{23} & 0 \\
 g_{22} & g_{23} & 0 & 0 \\
 g_{23} & 0 & 0 & 0 \\
}
$ & $
\pmatrx{
 0 & 0 & 0 & 0 \\
 0 & 0 & 0 & 0 \\
 0 & 0 & 0 & 0 \\
 0 & 0 & 0 & 0 \\
}
$ & $
\pmatrx{
 h_{11} & h_{12} & h_{22} & h_{23} \\
 h_{12} & h_{22} & h_{23} & 0 \\
 h_{22} & h_{23} & 0 & 0 \\
 h_{23} & 0 & 0 & 0 \\
}
$ & $\det h\neq 0$\smallskip\\

$\widetilde{A}_{3,4}$ & $
\pmatrx{
 e_1 & e_2 & e_3 & e_4 \\
 e_2 & e_4 & e_4 & 0 \\
 e_3 & e_4 & 0 & 0 \\
 e_4 & 0 & 0 & 0 \\
}
$ & $
\pmatrx{
 g_{11} & g_{12} & g_{13} & g_{23} \\
 g_{12} & g_{23} & g_{23} & 0 \\
 g_{13} & g_{23} & 0 & 0 \\
 g_{23} & 0 & 0 & 0 \\
}
$ & $
\pmatrx{
 0 & 0 & 0 & 0 \\
 0 & 0 & 0 & 0 \\
 0 & 0 & 0 & 0 \\
 0 & 0 & 0 & 0 \\
}
$ & $
\pmatrx{
 h_{11} & h_{12} & h_{13} & h_{23} \\
 h_{12} & h_{23} & h_{23} & 0 \\
 h_{13} & h_{23} & 0 & 0 \\
 h_{23} & 0 & 0 & 0 \\
}
$ & $\det h \neq 0$\smallskip\\

$\Alg_4$ & $
\pmatrx{
 0 & 0 & 0 & e_1 \\
 0 & 0 & 0 & e_2 \\
 0 & 0 & 0 & e_3 \\
 0 & 0 & 0 & e_4 \\
}
$ & $
\pmatrx{
 g_{11} & g_{12} & g_{13} & g_{14} \\
 g_{12} & g_{22} & g_{23} & g_{24} \\
 g_{13} & g_{23} & g_{33} & g_{34} \\
 g_{14} & g_{24} & g_{34} & g_{44} \\
}
$ & $
\pmatrx{
 0 & 0 & 0 & f_{14} \\
 0 & 0 & 0 & f_{24} \\
 0 & 0 & 0 & f_{34} \\
 -f_{14} & -f_{24} & -f_{34} & 0 \\
}
$ & $
\pmatrx{
 0 & 0 & 0 & 0 \\
 0 & 0 & 0 & 0 \\
 0 & 0 & 0 & 0 \\
 0 & 0 & 0 & h_{44} \\
}
$ & non-abelian\smallskip\\

\hline

\end{tabular}
\end{table}
}

\newpage

{%\scriptsize
\footnotesize
\begin{table}[h!]
\caption{Classification of bilinear forms associated with all four-dimensional non-transitive Novikov algebras for which the associated Lie algebra is non-abelian and nilpotent.}\label{tab6}
\begin{tabular}{ c  c  c  c  c  c }
\hline type & charact. matrix & $g$ & $f$ & $h$ & comments \\\hline\\[-10pt]

$N^{\mathfrak{h}_1}_{19}$ & $
\pmatrx{
 e_1 & e_2+e_3 & e_3 & e_4 \\
 e_2 & 0 & 0 & 0 \\
 e_3 & 0 & 0 & 0 \\
 e_4 & 0 & 0 & 0 \\
}
$ & $
\pmatrx{
 g_{11} & g_{12} & g_{13} & g_{14} \\
 g_{12} & 0 & 0 & 0 \\
 g_{13} & 0 & 0 & 0 \\
 g_{14} & 0 & 0 & 0 \\
}
$ & $
\pmatrx{
 0 & 0 & 0 & 0 \\
 0 & 0 & 0 & 0 \\
 0 & 0 & 0 & 0 \\
 0 & 0 & 0 & 0 \\
}
$ & $
\pmatrx{
 h_{11} & h_{12} & 0 & h_{14} \\
 h_{12} & 0 & 0 & 0 \\
 0 & 0 & 0 & 0 \\
 h_{14} & 0 & 0 & 0 \\
}
$ & $\det g =0$\smallskip\\

$N^{\mathfrak{h}_1}_{20}$ & $
\pmatrx{
 e_1 & e_2+e_3 & e_3 & e_4 \\
 e_2 & e_3 & 0 & 0 \\
 e_3 & 0 & 0 & 0 \\
 e_4 & 0 & 0 & 0 \\
}
$ & $
\pmatrx{
 g_{11} & g_{12} & g_{22} & g_{14} \\
 g_{12} & g_{22} & 0 & 0 \\
 g_{22} & 0 & 0 & 0 \\
 g_{14} & 0 & 0 & 0 \\
}
$ & $
\pmatrx{
 0 & 0 & 0 & 0 \\
 0 & 0 & 0 & 0 \\
 0 & 0 & 0 & 0 \\
 0 & 0 & 0 & 0 \\
}
$ & $
\pmatrx{
 h_{11} & h_{12} & 0 & h_{14} \\
 h_{12} & 0 & 0 & 0 \\
 0 & 0 & 0 & 0 \\
 h_{14} & 0 & 0 & 0 \\
}
$ & $\det g =0$\smallskip\\

$N^{\mathfrak{h}_1}_{21}$ & $
\pmatrx{
 e_1 & e_2+e_3 & e_3 & e_4 \\
 e_2 & e_4 & 0 & 0 \\
 e_3 & 0 & 0 & 0 \\
 e_4 & 0 & 0 & 0 \\
}
$ & $
\pmatrx{
 g_{11} & g_{12} & g_{13} & g_{22} \\
 g_{12} & g_{22} & 0 & 0 \\
 g_{13} & 0 & 0 & 0 \\
 g_{22} & 0 & 0 & 0 \\
}
$ & $
\pmatrx{
 0 & 0 & 0 & 0 \\
 0 & 0 & 0 & 0 \\
 0 & 0 & 0 & 0 \\
 0 & 0 & 0 & 0 \\
}
$ & $
\pmatrx{
 h_{11} & h_{12} & 0 & h_{22} \\
 h_{12} & h_{22} & 0 & 0 \\
 0 & 0 & 0 & 0 \\
 h_{22} & 0 & 0 & 0 \\
}
$ & $\det g =0$\smallskip\\

$N^{\mathfrak{h}_1}_{22}$ & $
\pmatrx{
 e_1 & e_2+e_3 & e_3 & e_4 \\
 e_2 & 0 & 0 & e_3 \\
 e_3 & 0 & 0 & 0 \\
 e_4 & e_3 & 0 & 0 \\
}
$ & $
\pmatrx{
 g_{11} & g_{12} & g_{24} & g_{14} \\
 g_{12} & 0 & 0 & g_{24} \\
 g_{24} & 0 & 0 & 0 \\
 g_{14} & g_{24} & 0 & 0 \\
}
$ & $
\pmatrx{
 0 & 0 & 0 & 0 \\
 0 & 0 & 0 & 0 \\
 0 & 0 & 0 & 0 \\
 0 & 0 & 0 & 0 \\
}
$ & $
\pmatrx{
 h_{11} & h_{12} & 0 & h_{14} \\
 h_{12} & 0 & 0 & 0 \\
 0 & 0 & 0 & 0 \\
 h_{14} & 0 & 0 & 0 \\
}
$ & \smallskip\\

$N^{\mathfrak{h}_1}_{23}$ & $
\pmatrx{
 e_1 & e_2+e_3 & e_3 & e_4 \\
 e_2 & e_4 & 0 & e_3 \\
 e_3 & 0 & 0 & 0 \\
 e_4 & e_3 & 0 & 0 \\
}
$ & $
\pmatrx{
 g_{11} & g_{12} & g_{24} & g_{22} \\
 g_{12} & g_{22} & 0 & g_{24} \\
 g_{24} & 0 & 0 & 0 \\
 g_{22} & g_{24} & 0 & 0 \\
}
$ & $
\pmatrx{
 0 & 0 & 0 & 0 \\
 0 & 0 & 0 & 0 \\
 0 & 0 & 0 & 0 \\
 0 & 0 & 0 & 0 \\
}
$ & $
\pmatrx{
 h_{11} & h_{12} & 0 & h_{22} \\
 h_{12} & h_{22} & 0 & 0 \\
 0 & 0 & 0 & 0 \\
 h_{22} & 0 & 0 & 0 \\
}
$ & \smallskip\\

$N^{\mathfrak{h}_1}_{24}$ & $
\pmatrx{
 e_1 & e_2+e_3 & e_3 & e_4 \\
 e_2 & e_3 & 0 & 0 \\
 e_3 & 0 & 0 & 0 \\
 e_4 & 0 & 0 & e_3 \\
}
$ & $
\pmatrx{
 g_{11} & g_{12} & g_{44} & g_{14} \\
 g_{12} & g_{44} & 0 & 0 \\
 g_{44} & 0 & 0 & 0 \\
 g_{14} & 0 & 0 & g_{44} \\
}
$ & $
\pmatrx{
 0 & 0 & 0 & 0 \\
 0 & 0 & 0 & 0 \\
 0 & 0 & 0 & 0 \\
 0 & 0 & 0 & 0 \\
}
$ & $
\pmatrx{
 h_{11} & h_{12} & 0 & h_{14} \\
 h_{12} & 0 & 0 & 0 \\
 0 & 0 & 0 & 0 \\
 h_{14} & 0 & 0 & 0 \\
}
$ & \smallskip\\

$N^{\mathfrak{h}_1}_{25}$ & $
\pmatrx{
 e_1 & e_2+e_3 & e_3 & e_4 \\
 e_2 & 0 & 0 & 0 \\
 e_3 & 0 & 0 & 0 \\
 e_4 & 0 & 0 & e_3 \\
}
$ & $
\pmatrx{
 g_{11} & g_{12} & g_{44} & g_{14} \\
 g_{12} & 0 & 0 & 0 \\
 g_{44} & 0 & 0 & 0 \\
 g_{14} & 0 & 0 & g_{44} \\
}
$ & $
\pmatrx{
 0 & 0 & 0 & 0 \\
 0 & 0 & 0 & 0 \\
 0 & 0 & 0 & 0 \\
 0 & 0 & 0 & 0 \\
}
$ & $
\pmatrx{
 h_{11} & h_{12} & 0 & h_{14} \\
 h_{12} & 0 & 0 & 0 \\
 0 & 0 & 0 & 0 \\
 h_{14} & 0 & 0 & 0 \\
}
$ & $\det g =0$\smallskip\\

$N^{\mathfrak{h}_2}_{15}$ & $
\pmatrx{
 e_1 & e_2+e_3 & e_3+e_4 & e_4 \\
 e_2 & 0 & 0 & 0 \\
 e_3 & 0 & 0 & 0 \\
 e_4 & 0 & 0 & 0 \\
}
$ & $
\pmatrx{
 g_{11} & g_{12} & g_{13} & g_{14} \\
 g_{12} & 0 & 0 & 0 \\
 g_{13} & 0 & 0 & 0 \\
 g_{14} & 0 & 0 & 0 \\
}
$ & $
\pmatrx{
 0 & 0 & 0 & 0 \\
 0 & 0 & 0 & 0 \\
 0 & 0 & 0 & 0 \\
 0 & 0 & 0 & 0 \\
}
$ & $
\pmatrx{
 h_{11} & h_{12} & 0 & 0 \\
 h_{12} & 0 & 0 & 0 \\
 0 & 0 & 0 & 0 \\
 0 & 0 & 0 & 0 \\
}
$ & $\det g =0$\smallskip\\

$N^{\mathfrak{h}_2}_{16}$ & $
\pmatrx{
 e_1 & e_2+e_3 & e_3+e_4 & e_4 \\
 e_2 & e_4 & 0 & 0 \\
 e_3 & 0 & 0 & 0 \\
 e_4 & 0 & 0 & 0 \\
}
$ & $
\pmatrx{
 g_{11} & g_{12} & g_{13} & g_{22} \\
 g_{12} & g_{22} & 0 & 0 \\
 g_{13} & 0 & 0 & 0 \\
 g_{22} & 0 & 0 & 0 \\
}
$ & $
\pmatrx{
 0 & 0 & 0 & 0 \\
 0 & 0 & 0 & 0 \\
 0 & 0 & 0 & 0 \\
 0 & 0 & 0 & 0 \\
}
$ & $
\pmatrx{
 h_{11} & h_{12} & 0 & 0 \\
 h_{12} & 0 & 0 & 0 \\
 0 & 0 & 0 & 0 \\
 0 & 0 & 0 & 0 \\
}
$ & $\det g =0$\smallskip\\

$N^{\mathfrak{h}_2}_{17}(\kappa)$ & $
\pmatrx{
e_ 1 & e_ 2+e_ 3 & e_ 3+e_ 4 & e_ 4 \\
 e_ 2 & 2 e_ 3+\kappa  e_ 4 & e_ 4 & 0 \\
 e_ 3 & e_ 4 & 0 & 0 \\
 e_ 4 & 0 & 0 & 0 \\
}
$ & $
\pmatrx{
 g_{11} & g_{12} & g_{13} & g_{23} \\
 g_{12} & g_{22} & g_{23} & 0 \\
 g_{13} & g_{23} & 0 & 0 \\
 g_{23} & 0 & 0 & 0 \\
}
$ & $
\pmatrx{
 0 & 0 & 0 & 0 \\
 0 & 0 & 0 & 0 \\
 0 & 0 & 0 & 0 \\
 0 & 0 & 0 & 0 \\
}
$ & $
\pmatrx{
h_{11} & h_{12} & 0 & 0 \\
 h_{12} & 0 & 0 & 0 \\
 0 & 0 & 0 & 0 \\
 0 & 0 & 0 & 0 \\
}
$ & $\begin{array}{l}g_{13}=\\\frac{1}{2}g_{22}+\frac{1-\kappa}{2} g_{23}\end{array}$\smallskip\\

\hline
\end{tabular}
\end{table}
}

\newpage

%\footnotesize

\bigskip\bigskip

\end{document}